\documentclass[reqno,11pt]{amsart}

\usepackage[numbers,sort&compress]{natbib}
\usepackage{color}
\usepackage{graphicx}
\usepackage{tikz}
\usepackage{rotating}
\usepackage{geometry}
\usepackage{mathtools,amssymb,amsthm,mathrsfs,color,paralist,graphicx,float}
\usepackage[colorlinks,
linkcolor=blue,
anchorcolor=green,
citecolor=blue, 
]{hyperref}

\newtheorem{theorem}{Theorem}[section]
\newtheorem{lemma}[theorem]{Lemma}
\newtheorem{proposition}[theorem]{Proposition}
\newtheorem{corollary}[theorem]{Corollary}
\newtheorem{remark}[theorem]{Remark}

\newtheorem{definition}[theorem]{Definition}

\newcommand{\C}{{\mathbb C}}

\newcommand{\Q}{{\mathbb Q}}
\newcommand{\R}{{\mathbb R}}
\newcommand{\T}{{\mathbb T}}

\newcommand{\Z}{{\mathbb Z}}


\catcode`\@=12

\def\Empty{}
\newcommand\oplabel[1]{
	\def\OpArg{#1} \ifx \OpArg\Empty {} \else
	\label{#1}
	\fi}

\newcommand{\key}[1]{\par\noindent Keywords: #1}

\begin{document}
	\title[Cantor spectrum]{Cantor Spectrum via a Reducibility-Duality Bridge for the Mosaic Almost Mathieu Operator}
	\author{Jiawei He}
	\address{
		Fujian Key Laboratory of Financial Information Processing, Putian University\\
		Fujian Putian, 351100, P.R. China} \email{hermit\_well@163.com}
	\author{Yuan Shan}
	\address{Corresponding author:
		Department of Mathematics, Nanjing Audit University, Nanjing 211815, China}
	\email{shanyuan@nau.edu.cn}
	
	\author{Yongjian Wang}
	\address{School of Mathematics and Statistics, Nanjing University of Science and Technology, Nanjing 210094, China}\email{wangyongjian@amss.ac.cn}
	%\date{\today}
	\begin{abstract}
		We study the mosaic Almost Mathieu operator, a quasiperiodic model that naturally admits a singular strip-Jacobi representation. Using Aubry duality for singular strip Jacobi operators and establishing the relation between the integrated density of states and the fibered rotation number in this setting, we reduce the spectral analysis to $\mathrm{SL}(2,\mathbb{R})$ cocycles. As a consequence, combining Aubry duality, reducibility theory, and a monotonicity-based cone-field argument, we prove that the spectrum is a Cantor set for every Diophantine frequency and every nonzero coupling.
	\end{abstract}

	\setcounter{tocdepth}{1}

	\maketitle
	%\linenumbers

	\key{Cantor spectrum; Reducibility; Rotation number; Aubry duality. }

	\section{Introduction}
	We consider one-dimensional discrete Schr\"odinger 
	operators on $\ell^2(\Z)$ of the form: 
	\begin{equation}
		[\mathcal{H}_{v, \alpha,\theta}u](n):=u(n+1)+u(n-1)+v(\theta+n\alpha)u(n),\ \forall n\in\Z,\nonumber
	\end{equation}
	where $\theta\in\mathbb{T}:=\mathbb{R}/\mathbb{Z}$ is the phase, $v\in C^{0}(\mathbb{T},\mathbb{R})$ is the sampling function, and $\alpha\in\mathbb{R}\setminus\mathbb{Q}$ is the frequency.
	The spectrum of $\mathcal{H}_{v, \alpha, \theta}$, denoted $\sigma(\mathcal{H}_{v,\alpha,\theta})$, is a compact subset of $\mathbb{R}$ independent of $\theta$. The bounded connected components of the resolvent set are called spectral gaps.
	We say that a Schr\"odinger operator has Cantor spectrum, if its spectrum $\sigma(\mathcal{H}_{v,\alpha,\theta})$ is a Cantor set. 
	\subsection{History and main result}
	In the 1980s, a wave of intense research activity focused on almost-periodic operators swept the mathematical physics community, notably initiated by Simon's influential work  \cite{Simon1982}.
	The central and most extensively studied special case is the Almost Mathieu operator (AMO), which takes the following form: 
	$$
	\left[\mathcal{H}_{\lambda, \alpha, \theta} u\right]_n=u_{n+1}+u_{n-1}+2 \lambda \cos 2 \pi(\theta+n \alpha) u_n, \quad \forall n \in \mathbb{Z} .
	$$
	Originally introduced by Peierls~\cite{Peierls73} to model electrons in two-dimensional lattices 
	under homogeneous magnetic fields~\cite{MagField55}, 
	the AMO subsequently proved fundamental to Thouless et al.'s theory of the integer quantum Hall effect~\cite{TKNN}. In recent years, the AMO has attracted extensive research interest owing to its profound quantum physical background. The richness of its spectral theory, abundance of anomalous phenomena, and remarkable analytical tractability further establish it as a cornerstone for mathematical conjectures and counterexamples.
	
	In this paper, we consider the following almost-periodic  operator on $\ell^2(\Z)$,
	\begin{equation}\label{eq1}
		[\mathcal{H}_{\theta}u](n):=u(n+1)+u(n-1)+v(\theta,n)u(n)
	\end{equation}
	where 
	\begin{equation}\label{potential}
		v(\theta, n)=\left\{\begin{array}{ccc}
			0 & n=1 & \bmod 2, \\
			2 \lambda \cos 2 \pi(n \alpha+\theta) & n=0 & \bmod 2 .
		\end{array}\right.
	\end{equation}
	Without loss of generality, assume $\lambda>0$. Originally introduced by Wang, Xia, Zhang, et al.~\cite{WangXXZZ23}, this model 
	features site-dependent potentials (distinct values at even versus odd lattice sites) and 
	was named the \textit{mosaic AMO operators} by the authors.  The mosaic AMO demonstrates 
	multiple mobility edges with self-duality breaking, extending the conceptual 
	framework first established in the Su--Schrieffer--Heeger model~\cite{SSH88} 
	and later observed in driven conformal field theory~\cite{wen21} among other contexts. The framework has already stimulated numerous physical investigations~\cite{Gonçalves23,Longhi24,liu21,zxc23,liu22}, with recent experimental realization achieved by means of hybrid integrated photonic devices~\cite{chang,Gao2024}.
	
	The mosaic AMO has a quasiperiodic potential on the even sites and zero potential on the odd sites, leading to a singular strip Jacobi representation.
	Recently, the spectral properties of the mosaic model (\ref{eq1}) were comprehensively analyzed by Wang, Xia, You, Zheng, and Zhou~\cite{exactm}. Their work specifically revealed the existence of mobility edges---critical energies separating absolutely continuous and pure point spectral components. While the spectral type of the model has been extensively investigated in recent years, the topological structure of its spectrum remained unknown.

	The frequency $\alpha$ is called Diophantine (denoted by $DC(\gamma, \tau)$), if there exist $\gamma, \tau>0$, such that
	$$
	\inf _{j \in \mathbb{Z}}|n \alpha-j| \geq \frac{\gamma}{|n|^\tau} \quad \forall n \in \mathbb{Z} \backslash\{0\}.
	$$
	We are now prepared to present our main result.
	\begin{theorem}\label{5-15-theorem1.1}
		If $\alpha\in DC(\gamma,\tau)$ and $\lambda\neq0$, then the spectrum of $\mathcal{H}_{\theta}$ is a Cantor set.
	\end{theorem}

	\begin{remark}
		Theorem \ref{5-15-theorem1.1} establishes the Cantor property of the spectrum. A Dry Ten Martini statement, understood as the opening of every gap allowed by the gap-labeling theorem, fails for the mosaic AMO because its central allowed gap is necessarily collapsed.
		
		Let $\Sigma$ denote the  spectrum and let $n$ be the integrated density of states normalized per lattice site. For the underlying translation on $\mathbb{T}\times\mathbb{Z}_{2}$, the gap-labeling group is
		$\alpha\mathbb{Z}+\frac{1}{2}\mathbb{Z};$
		see \cite[Theorem~6.12]{DFgap}. In particular, $\frac{1}{2}$ is an allowed internal label. On the other hand, the sequence
		\begin{equation*}
			u(2m)=0,\qquad u(2m+1)=(-1)^{m},\qquad m\in\mathbb{Z},
		\end{equation*}
		satisfies $\mathcal{H}_{\theta}u=0$ for every $\theta$. Its normalized truncations form a Weyl sequence, so $0\in\Sigma$. Moreover, the unitary map $u(n)\mapsto(-1)^{n}u(n)$ conjugates $\mathcal{H}_{\theta}$ to $-\mathcal{H}_{\theta+\frac{1}{2}}$. Thus the density-of-states measure is symmetric about zero, and continuity of the IDS gives
		\begin{equation*}
			n(-E)=1-n(E),\qquad n(0)=\frac{1}{2}.
		\end{equation*}
		Since $0\in\operatorname{supp}(\mathrm{d}n)=\Sigma$, for every $E>0$ we have
		\begin{equation*}
			2n(E)-1=n(E)-n(-E)=\int_{(-E,E]}\mathrm{d}n>0.
		\end{equation*}
		Consequently,
		\begin{equation*}
			n^{-1}\Big(\Big\{\frac{1}{2}\Big\}\Big)=\{0\},
		\end{equation*}
		which shows that the allowed gap with label $\frac{1}{2}$ is collapsed.
		
		For comparison, the standard AMO with irrational frequency has gap-labeling group $\mathbb{Z}+\alpha\mathbb{Z}$, which does not contain $\frac{1}{2}$; see \cite[Theorem~6.2]{DFgap}. Hence its zero-energy spectral point does not represent a collapsed allowed gap in this sense. Our proof opens sufficiently many gaps to establish the Cantor property of the mosaic spectrum; whether every other allowed label corresponds to an open gap is not settled here.
	\end{remark}
	
	Let us  turn to  the almost Mathieu operators. The AMO has long been conjectured  to possess Cantor spectrum for irrational $\alpha$ (known as Azbel conjecture \cite{Azbel1964} and Ten Martini Problem \cite{Simon2000}). This conjecture was definitively resolved by Avila and Jitomirskaya \cite{AJ08}. Significant progress on the Ten Martini Problem has been driven by many outstanding works, including foundational contributions by Bellissard--Simon~\cite{BS1982}, Sinai~\cite{S87}, Choi--Elliott--Yui~\cite{CEY1990}, and Puig~\cite{P04}, among others.
	Inspired by Hofstadter's numerical results~\cite{Ho1976}, Kac posited that all spectral gaps of the AMO are open---a stronger formulation termed the ``Dry Ten Martini Problem''~\cite{Simon1982}. This stronger version has been resolved for the noncritical case ($\lambda \neq 1$) by Avila--You--Zhou~\cite{AYZ}.
	Recently, Ge-Jitomirskaya-You \cite{GJY} proved that Cantor spectrum is also present for AMO under analytic perturbations. Their proof method establishes Cantor spectrum without relying on almost Mathieu symmetry, self-duality, or low-degree potentials. It applies universally to all irrational frequencies and works for a large open set of analytic potentials—termed Type I potentials.
	
	Cantor spectrum was conjectured to be a generic phenomenon for one-dimensional almost-periodic operators by Simon~\cite{Simon1982}. The analysis becomes significantly more challenging for Schr\"odinger operators with general analytic potentials. In the region
	of positive Lyapunov exponent, Goldstein and Schlag \cite{GS11} proved that for
	any analytic potential, the spectrum is a Cantor set for almost every frequency. In the perturbative regime, Eliasson \cite{E92} proved that for a fixed
	Diophantine frequency, the spectrum is a Cantor set for generic such potentials in the standard analytic topology. This result was subsequently extended by Puig~\cite{puig06} to quasi-periodic Schr\"odinger operators with fixed Diophantine frequencies and potentials in the non-perturbative regime.

	We further survey results concerning Cantor spectrum beyond the analytic category. In the finitely differentiable setting, Eliasson's perturbative theorem was extended by Cai and Ge \cite{CAIGE22}, while large-coupling Cantor spectrum results were obtained for smooth cosine-type potentials by Sinai \cite{S87}, Wang and Zhang \cite{WZ}, and Ge, Wang and Xu \cite{GWX}. In Gevrey classes, Hou, Shan and You \cite{HSY} developed a constructive scheme based on reducibility and gap-opening arguments, leading to explicit examples of quasiperiodic Schrödinger operators with Cantor spectrum. Related developments and refinements were subsequently obtained in \cite{HHSY,HZ}. In the continuous topology, Avila, Bochi and Damanik \cite{ABD09} established generic Cantor spectrum through gap-opening techniques. Related developments have also appeared in the setting of Sturmian Hamiltonians \cite{BBL} and almost-periodic CMV matrices \cite{DL,LDZ}. These results provide strong evidence that Cantor spectrum is a pervasive phenomenon in one-dimensional almost-periodic spectral theory.
	
	The mosaic model differs substantially from the quasiperiodic Schr\"odinger operators considered in the works above. Indeed, it may be naturally regarded as a Schrödinger operator generated by a product dynamical system \cite{DFG}, whose site-dependent sampling structure gives rise to a singular strip-Jacobi representation. Moreover, the exact self-duality that underlies many classical results for the almost Mathieu operator and related models is absent. Consequently, the mosaic AMO falls outside the reach of the standard analytic quasiperiodic framework, and neither the classical Aubry-duality approach nor the existing gap-opening arguments can be applied in a straightforward manner.
	
	The goal of this paper is to establish the Cantor property of the mosaic AMO spectrum. We combine the Aubry duality for singular strip Jacobi operators developed in \cite{wang2022} with reducibility theory and the relation between the integrated density of states and the fibered rotation number. A monotonicity-based cone-field argument along the original energy parameter then yields sufficiently many open gaps to prove that the spectrum is a Cantor set.
	We expect that these ideas may also prove useful in the study of spectral topology for other almost-periodic models with singular structures, mobility edges, or the absence of classical self-duality.

	\subsection{Major Difficulties}
	
	The proof of Theorem \ref{5-15-theorem1.1} involves three main difficulties.
	
	First, the mosaic model does not have the standard AMO self-duality. We use the singular strip Jacobi representation and the Aubry-duality framework of \cite{wang2022}. Since the hopping matrix is singular, the usual transfer-matrix construction for nonsingular strip Jacobi operators does not apply directly. We work instead with the return cocycle $(2\alpha,D_{E}^{v})$ and, for $E\neq0$, the dual cocycle $(2\alpha,\hat{A}^{E})$. Their exact relations with almost Mathieu cocycles allow us to apply the available reducibility results.
	
	Second, the relation between the integrated density of states and the fibered rotation number requires special care. The second diagonal entry of the dual potential is identically zero, and eliminating the second component introduces a singularity at $E=0$. We analyze finite-volume eigenvalue counts separately on the positive and negative energy axes and establish the corresponding rotation-number formulas.
	
	Third, a parabolic reduction at a single energy must lead to an open gap in the original energy parameter. We verify parameter monotonicity for the return cocycle and for an iterate of the dual cocycle. The cone criterion then gives uniform hyperbolicity on a one-sided interval adjacent to the resonant energy, which yields the gap opening needed for the Cantor-spectrum argument.

	\subsection{Outline of the paper} 
	The remainder of this work is structured as follows. Section 2 presents preliminary definitions and foundational results. Section 3 introduces the framework and key notions of Aubry duality for strip Jacobi operators. Section 4 establishes the equivalence between the integrated density of states and the fibered rotation number for the singular strip Jacobi operator. Section 5 develops reducibility theorems in the subcritical regime and a monotonicity-based cone-field criterion for gap opening. Finally, Section 6 completes the proof of Theorem \ref{5-15-theorem1.1} by contradiction, unifying reducibility, Aubry duality and analytic spectral methods.

	\section{Preliminaries}

	Denote by $C^{\omega}_{h}(\mathbb{T}, *)$ the set of $*$-valued functions ($*$ could be $\mathbb{R}$, $\mathrm{sl}(2,\mathbb{C})$ or $\mathrm{SL}(2,\mathbb{C})$) admitting a holomorphic extension to $\{|\Im \theta|<h\}$. For any $F\in C^{\omega}_{h}(\mathbb{T}, *)$, it has a Fourier expansion
	
	\begin{equation*}
		F(\theta) = \sum_{k \in \mathbb{Z}} \widehat{F}(k) \mathrm{e}^{2\pi \mathrm{i} k \theta}, \quad \widehat{F}(k) := \int_{\mathbb{T}} F(\theta) \mathrm{e}^{-2\pi \mathrm{i} k \theta} \mathrm{d} \theta.
	\end{equation*}
	
	Given $h \geqslant 0$, we introduce the Wiener norm of $F $ as
	\begin{equation*}
		\|F\|_{h}:=\sum_{k\in\mathbb{Z}} |\widehat{F}(k)| \mathrm{e}^{2\pi|k| h}.
	\end{equation*}
	
	\iffalse
	\subsection{Continued Fraction Expansion.}\quad 
	When $\theta\in\mathbb{R}$, we also set $\|\theta\|_{\mathbb{T}}=\inf_{j\in\mathbb{Z}}|\theta-j|$. Let $\alpha\in(0,1)$ be irrational, $a_{0}=0$ and $b_{0}=\alpha$. Inductively, for $k\ge1$, we define
	\begin{equation*}\quad
		a_{k}=\lfloor b_{k-1}^{-1}\rfloor,\ b_{k}=b_{k-1}^{-1}-a_{k},
	\end{equation*}
	Let $p_{0}=0$, $p_{1}=1$, $q_{0}=1$, $q_{1}=a_{1}$. We define inductively $p_{k}=a_{k}p_{k-1}+p_{k-1}$, $q_{k}=a_{k}q_{k-1}+q_{k-2}$. Then $(q_{n})_{n}$ is the sequence of denominators of the best rational approximations of $\alpha$, since we have $\|k\alpha\|_{\mathbb{T}}\ge\|q_{n-1}\alpha\|_{\mathbb{T}}$, $\forall\ 1\le k<q_{n}$, and 
	\begin{equation*}\quad
		\frac{1}{2q_{n+1}}\le\|q_{n}\alpha\|_{\mathbb{T}}\le\frac{1}{q_{n+1}}.
	\end{equation*}

	\begin{lemma}\cite{AJ}\label{ten}
		Let $\alpha\in\mathbb{R}\backslash\mathbb{Q}$, $x\in\mathbb{R}$ and $0\le l_0\le q_n-1$ be such that $$|\sin\pi(x+l_0\alpha)|=\inf_{0\le l\le q_n-1}|\sin\pi(x+l\alpha)|,$$ then for some absolute constant $C>0$,
		$$-C\ln q_n\le\sum_{0\le l\le q_n-1,l\neq l_0}\ln|\sin\pi(x+l\alpha)|+(q_n-1)\ln2\le C\ln q_n.$$
	\end{lemma}
	\fi
	
	\subsection{Cocycle, Lyapunov exponent} \label{co-ly}
	Let $X$ be a compact metric space and let $T:X\to X$ be a homeomorphism preserving an ergodic Borel probability measure $\nu$. For $A\in C^{0}(X,M(2,\mathbb{R}))$, the cocycle $(T,A)$ is the linear skew product
	\begin{equation*}
		\begin{aligned}
			(T,A):X\times\mathbb{R}^{2}&\to X\times\mathbb{R}^{2},\\
			(x,\phi)&\mapsto(Tx,A(x)\phi).
		\end{aligned}
	\end{equation*}
	Its forward iterates satisfy $(T,A)^{n}=(T^{n},A_{n})$, where $A_{0}(x)=\mathrm{id}$ and
	\begin{equation}\label{n-step transfer}
		A_{n}(x)=A(T^{n-1}x)\cdots A(Tx)A(x),\qquad n\geqslant1.
	\end{equation}
	Thus $A_{n}:X\to M(2,\mathbb{R})$. If $A(x)$ is invertible for every $x\in X$, we also define the negative iterates by
	\begin{equation*}
		A_{-n}(x)=A_{n}(T^{-n}x)^{-1},\qquad n\geqslant1.
	\end{equation*}
	When $A$ takes values in $\mathrm{SL}(2,\mathbb{R})$, so do all its iterates.
	
	The Lyapunov exponent is defined by
	\begin{equation*}
		L(T,A)=\lim_{n\to\infty}\frac{1}{n}\int_{X}\ln\|A_{n}(x)\|\,\mathrm{d}\nu(x).
	\end{equation*}
	The limit exists by subadditivity. The same definitions apply to complex matrix-valued cocycles, with $M(2,\mathbb{R})$ and $\mathbb{R}^{2}$ replaced by $M(2,\mathbb{C})$ and $\mathbb{C}^{2}$, respectively.

	In this paper, we will consider the following   cocycles:
	\begin{itemize}
		\item  $X=\mathbb{T}$ and $T=R_\alpha$, where  $R_{\alpha}\theta= \theta+\alpha$, then $(\alpha,A):=(R_{\alpha},A)$ is a  quasi-periodic cocycle.
		\item  $X=\mathbb{T}\times\mathbb{Z}_{2}$ and $T=T_{\alpha}$, $T_{\alpha} (\theta,n):= (\theta+\alpha,n+1)$, then $(T_{\alpha},A)$ defines an almost-periodic cocycle.
	\end{itemize}
	These dynamical  systems $(X,T)$ are uniquely ergodic if  $\alpha$ is irrational (Theorem 9.1 of \cite{mane2012ergodic}).
	
	The potential in \eqref{potential} can be generated by the translation $T_{\alpha}$ on $X=\mathbb{T}\times\mathbb{Z}_{2}$. From now on, $v$ denotes the sampling function on $X$ defined by
	\begin{equation*}
		v(\theta,0)=2\lambda\cos(2\pi\theta),\qquad v(\theta,1)=0.
	\end{equation*}
	Under $T_{\alpha}(\theta,j)=(\theta+\alpha,j+1)$, with the second coordinate taken modulo $2$, we have
	\begin{equation*}
		v\big(T_{\alpha}^{n}(\theta,0)\big)=
		\begin{cases}
			2\lambda\cos\big(2\pi(\theta+n\alpha)\big), & n\equiv0\pmod{2},\\
			0, & n\equiv1\pmod{2}.
		\end{cases}
	\end{equation*}
	Thus the operator \eqref{eq1} is generated by the sampling function $v$ and induces the Schr\"odinger cocycle $(T_{\alpha},S_{E}^{v})$, where
	\begin{equation*}
		S_{E}^{v}(\theta,j)=
		\begin{pmatrix}
			E-v(\theta,j) & -1\\
			1 & 0
		\end{pmatrix},\qquad (\theta,j)\in X.
	\end{equation*}
	Since $T_{\alpha}^{2}(\theta,0)=(\theta+2\alpha,0)$, the two-step return to $\mathbb{T}\times\{0\}$ defines the analytic quasiperiodic cocycle $(2\alpha,D_{E}^{v})$, with
	\begin{equation*}
		\begin{aligned}
			D_{E}^{v}(\theta)
			&=S_{E}^{v}(\theta+\alpha,1)S_{E}^{v}(\theta,0)\\
			&=
			\begin{pmatrix}
				E^{2}-1-2\lambda E\cos(2\pi\theta) & -E\\
				E-2\lambda\cos(2\pi\theta) & -1
			\end{pmatrix}.
		\end{aligned}
	\end{equation*}

	We say an $SL(2,\mathbb{R})$ cocycle $(T,A)$  is uniformly hyperbolic if, for every $x \in X$, there exists a continuous splitting $\mathbb{R}^2=E_{s}(x)\oplus E_{u}(x)$ such that for every $n\ge0$,
	\begin{equation*}\quad
		\begin{split}
			|A_{n}(x)v(x)|&\le Ce^{-cn}|v(x)|,\ v(x)\in E_{s}(x),\\
			|A_{-n}(x)v(x)|&\le Ce^{-cn}|v(x)|,\ v(x)\in E_{u}(x),
		\end{split}
	\end{equation*}
	for some constant $C,c>0$. Clearly, it holds that $A(x)E_{s}(x)=E_{s}(Tx)$ and $A(x)E_{u}(x)=E_{u}(Tx)$ for every $x\in X$, and if $(T,A)$ is uniformly hyperbolic, then $L(T,A)>0$.

	\subsection{Fibered rotation number.}
	Let $\mathbb S^{1}$ be the set of unit vectors of $\mathbb{R}^{2}$, consider a projective cocycle $F_{A}$ on $X\times\mathbb{S}^1$: $$(x,\phi)\mapsto (Tx,\frac{A(x)\phi}{\|A(x)\phi\|}).$$
	If $A\in C^{0}(X,\mathrm{SL}(2,\mathbb{R}))$  is homotopic to the
	identity,    then there exists
	a lift $\tilde{F}_{A}$ of $F_{A}$ to $X \times \mathbb{R}$ such that $\tilde{F}_{A}(x,\phi) = (Tx, \tilde{f}_{A}(x,\phi)) $  where $\tilde{f}_{A} :X \times \mathbb{R}\rightarrow\mathbb{R}$ is a continuous lift such that
	\begin{itemize}
		\item $\tilde{f}_{A}(x,\phi + 1) = \tilde{f}_{A}(x,\phi) + 1;$
		\item for every $x \in X, \tilde{f}_{A}(x, \cdot) : \R \rightarrow \R$ is a strictly increasing homeomorphism;
		\item if $\pi_2$ is the projection map  $X \times \mathbb{R}\rightarrow X \times \mathbb{S}^1:(x,\phi) \mapsto (x, e^{2\pi i\phi})$, then $F_{A} \circ \pi_2 = \pi_2 \circ \tilde{F}_{A}$.
	\end{itemize}
	Suppose that $(X,\nu,T)$ is uniquely ergodic. For a fixed continuous lift $\tilde{F}_{A}$, the limit
	\begin{equation*}
		\rho(T,A)=\lim_{n\to\infty}\frac{\tilde{f}_{A}^{n}(x,\phi)-\phi}{n}
	\end{equation*}
	exists and is independent of $(x,\phi)\in X\times\mathbb{R}$, where $\tilde{f}_{A}^{n}$ denotes the second component of $\tilde{F}_{A}^{n}$. If $X=\mathbb{T}$ and $T=R_{\alpha}$, any two continuous lifts differ by a constant integer. Hence this limit modulo $\mathbb{Z}$ is independent of the lift, and we denote the resulting fibered rotation number by $\rho(\alpha,A)$; see \cite{johnson1982rotation,herman1983methode}.
	
	For $X=\mathbb{T}\times\mathbb{Z}_{2}$, two continuous lifts may differ by different integers on the two connected components. Thus the rotation number need not be independent of the lift modulo $\mathbb{Z}$. For the Schr\"odinger cocycle $(T_{\alpha},S_{E}^{v})$, we fix the continuous lift by requiring
	\begin{equation*}
		0<\tilde{f}_{S_{E}^{v}}(x,0)<\frac{1}{2},\qquad x\in\mathbb{T}\times\mathbb{Z}_{2}.
	\end{equation*}
	This lift is uniquely determined because
	\begin{equation*}
		S_{E}^{v}(x)\begin{pmatrix}1\\0\end{pmatrix}
		=
		\begin{pmatrix}E-v(x)\\1\end{pmatrix}
	\end{equation*}
	lies in the open upper half-plane. Throughout the paper, $\rho(T_{\alpha},S_{E}^{v})$ denotes the real rotation number associated with this fixed lift. In particular,
	\begin{equation*}
		0\leqslant\rho(T_{\alpha},S_{E}^{v})\leqslant\frac{1}{2}.
	\end{equation*}
	This normalization is used in the rotation--IDS relation below.
	
	For the return cocycle $(2\alpha,D_{E}^{v})$, the lift obtained by composing the two one-step lifts has twice the rotation number of the original cocycle. Consequently,
	\begin{equation*}
		\rho(2\alpha,D_{E}^{v})\equiv2\rho(T_{\alpha},S_{E}^{v})\pmod{\mathbb{Z}}.
	\end{equation*}

	\iffalse
	The fibered rotation number is invariant under real conjugacies which are homotopic to the identity. In general, if the cocycles $(\alpha,A_1)$ is conjugated to $(\alpha,A_2)$: 
	$$B(\theta+\alpha)^{-1}A_1(\theta)B(\theta)=A_2(\theta),$$ and $B \in
	C^0(\T,$ $	P\mathrm{SL}(2,\R))$ has degree n (that is, it is homotopic to $\theta \mapsto R_{n\theta/2}$), where 
	\begin{equation*}
		R_{\phi}=\begin{pmatrix}\cos2\pi\phi& -\sin2\pi\phi\\ \sin2\pi\phi& \cos2\pi\phi\end{pmatrix},
	\end{equation*}
	then  we have
	\begin{equation}
		\rho({\alpha,A_{1}})=\rho(\alpha,A_{2})+\frac{1}{2}n\alpha \mod\ \mathbb{Z}.
	\end{equation}
	If furthermore $B \in C^0(\T,$ $\mathrm{SL}(2,\R))$ with  $\deg B=n \in\Z$, then we
	\begin{equation*}
		\rho({\alpha,A_{1}})=\rho(\alpha,A_{2})+n\alpha \mod \Z.
	\end{equation*}
	\fi
	
	\subsection{Dynamically  defined Schr\"odinger operators.} 
	Let $X$ be a compact metric space, $(X,\nu,T)$ be ergodic, and $v: X \rightarrow \R$ be continuous. The associated Schr\"odinger operator on $\ell^2(\mathbb{Z})$ is defined by:
	\begin{equation*}
		[\mathcal{H}_{x}u]_{n}=u_{n+1}+u_{n-1}+v(T^n x )u_{n},\ \ \forall x \in X. 
	\end{equation*}
	The spectrum $\sigma(\mathcal{H}_x)$ is a compact subset of $\mathbb{R}$ and, when $(X,T)$ is minimal, is independent of $x$ \cite{damanik2017schrodinger}.

	Any formal solution $u=(u_{n})_{n\in\mathbb{Z}}$ to $\mathcal{H}_{x}u = Eu$ satisfies the recurrence relation:
	\begin{equation*}
		\begin{pmatrix} u_{n+1} \\ u_n \end{pmatrix} = S_E^v(T^n x) \begin{pmatrix} u_n \\ u_{n-1} \end{pmatrix},
	\end{equation*}
	with the \textit{Schr\"odinger cocycle} 
	\begin{equation}\label{eq:cocycle}
		S_E^v(\cdot) = \begin{pmatrix} E - v(\cdot) & -1 \\ 1 & 0 \end{pmatrix}.
	\end{equation}
	If $(X,T)$ is minimal, then for every $x\in X$ and $E\in\mathbb{R}$, one has $E\notin\sigma(\mathcal{H}_{x})$ if and only if $(T,S_{E}^{v})$ is uniformly hyperbolic~\cite{johnson1986exponential}.

	For $x\in X$ and $\phi\in\ell^{2}(\mathbb{Z})$, let $\mu_{x}^{\phi}$ be the spectral measure of $\mathcal{H}_{x}$ associated with $\phi$, characterized by
	\begin{equation*}
		\langle\phi,f(\mathcal{H}_{x})\phi\rangle=\int_{\sigma(\mathcal{H}_{x})}f(E)\,\mathrm{d}\mu_{x}^{\phi}(E),\qquad f\in C^{0}\big(\sigma(\mathcal{H}_{x})\big).
	\end{equation*}
	In particular, $\mu_{x}^{\phi}(\mathbb{R})=\|\phi\|^{2}$. Let $e_{0}$ and $e_{1}$ be the standard unit vectors at sites $0$ and $1$, respectively. We use the normalized average
	\begin{equation*}
		\mathrm{d}\mu_{x}=\frac{1}{2}\big(\mathrm{d}\mu_{x}^{e_{0}}+\mathrm{d}\mu_{x}^{e_{1}}\big),
	\end{equation*}
	which is a probability measure.
	
	The integrated density of states (IDS) is defined by
	\begin{equation}\label{eq:ids}
		n(E)=\int_{X}\mu_{x}\big((-\infty,E]\big)\,\mathrm{d}\nu(x).
	\end{equation}
	By covariance of the operator family and $T$-invariance of $\nu$, replacing $\mu_{x}$ by $\mu_{x}^{e_{0}}$ in \eqref{eq:ids} gives the same function. Thus this definition agrees with the usual per-site IDS. For the ergodic Schr\"odinger family considered here, $n$ is nondecreasing and continuous; see \cite[Theorem~3.6]{DFgap}. The individual spectral measures $\mu_{x}$ may nevertheless have atoms.
	
	Recall that a spectral gap is a bounded connected component of $\mathbb{R}\setminus\sigma(\mathcal{H}_{x})$. The relation between the IDS and the fibered rotation number, with the normalized real lift understood, is recalled below.

	\begin{theorem}\cite{herman1983methode,Marx2017,johnson1986exponential,johnson1982rotation}
		\label{thm:rotation_ids}
		For quasiperiodic Schr\"odinger operators, the following properties hold:
		\begin{enumerate}
			\item \textbf{(Gap Labeling)} For each spectral gap $G_k(V) = (E_k^-(V), E_k^+(V))$, there exists unique $k \in \mathbb{Z}$ such that
			\begin{equation}\label{eq:gap_label}
				n(E) \equiv \langle k, \alpha \rangle\pmod{\mathbb{Z}} \quad \forall E \in G_k(V)
			\end{equation}
			
			\item \textbf{(Rotation-IDS Relation)} The fibered rotation number satisfies $\rho(T, S_E^v) \in \left[0, \frac{1}{2}\right]$ with
			\begin{equation}\label{eq:rho_ids}
				n(E) = 1 - 2\rho(T, S_E^v)
			\end{equation}
			for all $E \in \mathbb{R}$.
		\end{enumerate}

	\end{theorem}
	For the mosaic family \eqref{eq1}, the rotation--IDS relation \eqref{eq:rho_ids} also holds with $T=T_{\alpha}$ and the normalization fixed in Subsection 2.2; see \cite[Theorem~5.3]{DFgap}. Its gap-labeling group is $\alpha\mathbb{Z}+\frac{1}{2}\mathbb{Z}$, so its allowed internal labels are the elements of
	$\Big(\alpha\mathbb{Z}+\frac{1}{2}\mathbb{Z}\Big)\cap(0,1);$
	see \cite[Theorem~6.12]{DFgap}. 
	For either family, continuity and monotonicity of the IDS imply that the inverse image under $n$ of an allowed internal label is a nonempty compact interval. We call the corresponding gap open if this interval has nonempty interior, and collapsed if it consists of a single point. Since the support of the density-of-states measure equals the spectrum, in the open case the interior of this interval is precisely a bounded connected component of the complement of the spectrum.

	\iffalse
	By Thouless formula \cite{AS}, we also have the following relation between the integrated density of states and the Lyapunov exponent:
	$$L(T,S_{E}^{v})=\int\ln|E-E'|d\mathrm{n}(E').$$	
	\fi
	\subsection{Global theory of one frequency quasi-periodic cocycle}
	Let us make a short review of Avila's global theory of one-frequency quasi-periodic cocycles \cite{A4}. Suppose that $D \in C^\omega(\mathbb{T}, M(2, \mathbb{C}))$ admits a holomorphic extension to $\{|\Im \theta|<h\}$ and assume that $L(\alpha,D_{\varepsilon})$ is finite for every $|\varepsilon|<h$. Then for $|\epsilon|<h$, we define $D_\epsilon \in C^\omega(\mathbb{T}, M(2, \mathbb{C}))$ by $D_\epsilon(\cdot)=D(\cdot+i \epsilon)$, and define the  acceleration of ($\alpha, D_{\varepsilon}$) as follows
	
	$$
	\omega\left(\alpha, D_{\varepsilon}\right)=\frac{1}{2 \pi} \lim _{h \rightarrow 0+} \frac{L\left(\alpha, D_{\varepsilon+h}\right)-L\left(\alpha, D_{\varepsilon}\right)}{h}.
	$$

	The acceleration was first introduced by Avila for analytic $\mathrm{SL}(2, \mathbb{C})$ cocycles \cite{A4}, and extended to analytic $M(2, \mathbb{C})$ cocycles by Jitomirskaya and Marx \cite{jitomirskaya2012analytic}. It follows from the convexity and continuity of the Lyapunov exponent that the acceleration is an upper semicontinuous function in parameter $\varepsilon$. The key property of the acceleration is that it is quantized:
	
	\begin{theorem}\label{5-13-theorem2.1}(Quantization of acceleration) \cite{A4,jitomirskaya2012analytic,jit2013} Suppose that $(\alpha, D) \in$ $(\mathbb{R} \backslash \mathbb{Q}) \times C^\omega\left(\mathbb{T}, M_2(\mathbb{C})\right)$ with $\det D(\theta)$ bounded away from zero on the strip $\{|\Im \theta|<$ $h\}$, then $\omega\left(\alpha, D_{\varepsilon}\right) \in \frac{1}{2} \mathbb{Z}$ in the strip. Moreover, if $D \in C^\omega(\mathbb{T}, \mathrm{SL}(2, \mathbb{C}))$, then $\omega\left(\alpha, D_{\varepsilon}\right) \in \mathbb{Z}$.
		
	\end{theorem}

	If $A$ takes values in $\mathrm{SL}(2, \mathbb{R})$, then $\varepsilon \mapsto L\left(\alpha, A_{\varepsilon}\right)$ is an even function. By convexity, $\omega(\alpha, A) \geq 0$. And if $\alpha \in \mathbb{R} \backslash \mathbb{Q}$, then $(\alpha, A)$ is uniformly hyperbolic if and only if $L(\alpha, A)>0$, and $\omega(\alpha, A)=0$. The cocycles in $\mathrm{SL}(2, \mathbb{R})$ which are not uniformly hyperbolic are classified into three regimes: subcritical, critical, and supercritical. Especially, $(\alpha, A)$ is said to be subcritical if $L(\alpha, A)=0, \omega(\alpha, A)=0$; the cocycle $(\alpha, A)$ is said to be supercritical if $L(\alpha, A)>0, \omega(\alpha, A)>0 ;$ otherwise $(\alpha, A)$ is critical.
	
	A cornerstone in Avila’s global theory is the “Almost Reducibility Conjecture” (ARC),
	which says that $(\alpha, A)$ is almost reducible if it is subcritical.
	
	\begin{theorem}\cite{A4,Ge2023}\label{ARC}
		Given $\alpha \in \mathbb{R} \backslash \mathbb{Q}$, and $A \in C^\omega(\mathbb{T}, \mathrm{SL}(2, \mathbb{R}))$, if $(\alpha,A)$ is subcritical, then it is almost reducible.
	\end{theorem}

	\section{Aubry duality in strip Jacobi operators}
	Aubry duality serves as a crucial tool in the spectral theory of quasi-periodic operators and has been used extensively in the analysis of localization-delocalization transitions \cite{AJ1,AJM, exactm,wang2022} and the Cantor spectrum problem \cite{P04,GJY}. Recent work \cite{wang2022} extends Aubry duality to strip Jacobi operators with singular hopping matrices. The operator defined in \eqref{eq1} can be precisely represented as such a strip Jacobi operator, enabling direct application of these extended duality results. Here we briefly outline the  framework and key notions of Aubry duality in strip Jacobi operators.
	
	Consider the strip Jacobi operator on $\ell^2(\mathbb{Z}, \mathbb{C}^2)$ defined by
	\begin{equation}\label{eq:strip_jacobi}
		[\mathcal{S}_{C,V,\theta} \boldsymbol{u}](n) := C \boldsymbol{u}(n + 1) + C^* \boldsymbol{u}(n - 1) + V(\theta + 2n\alpha) \boldsymbol{u}(n),
	\end{equation}
	where $C\in M_{2}(\mathbb{R})$ is the hopping matrix and $V\in C^{\omega}(\mathbb{T},M_{2}(\mathbb{R}))$ is the potential matrix satisfying
	$V(\theta)^{*}=V(\theta).$
	Here $^{*}$ denotes the conjugate transpose. The operator $\mathcal{S}_{C,V,\theta}$ is called singular when $\det C=0$.
	
	Let $V^{(k)}$ denote the Fourier coefficients of $V$, so that
	\begin{equation*}
		V(\theta)=\sum_{k\in\mathbb{Z}}V^{(k)}e^{2\pi ik\theta},\qquad V^{(k)}=\int_{\mathbb{T}}V(x)e^{-2\pi ikx}\,\mathrm{d}x.
	\end{equation*}
	For every $\theta\in\mathbb{T}$, define the dual operator on $\ell^{2}(\mathbb{Z},\mathbb{C}^{2})$ by
	\begin{equation*}
		\begin{aligned}
			[\widehat{\mathcal{S}}_{C,V,\theta}\boldsymbol{\psi}](n)
			&=\sum_{k\in\mathbb{Z}}V^{(k)}\boldsymbol{\psi}(n-k)\\
			&\quad+\big(Ce^{2\pi i(2n\alpha+\theta)}+C^{*}e^{-2\pi i(2n\alpha+\theta)}\big)\boldsymbol{\psi}(n),\qquad n\in\mathbb{Z}.
		\end{aligned}
	\end{equation*}
	Consider Hilbert space $\mathbb{H}=L^2(\T\times\Z,\C^2)$, which consists of functions $\boldsymbol{\Psi}=\boldsymbol{\Psi}(\theta,n)$ satisfying
	\begin{align*}
		\sum_{n\in\Z}\int_{\T}\|\boldsymbol{\Psi}(\theta,n)\|^2d\theta<\infty.
	\end{align*}
	The extensions of the singular operators $\mathcal{S}_{C,V,\theta}$ and their long-range duals $\widehat{\mathcal{S}}_{C,V,\theta}$ to $\mathbb{H}$ are given in terms of their {\em direct integrals}. The direct integral of the singular operator $\mathcal{S}_{C,V,\theta}$ is the operator $\mathbf{S}_{C,V}$, defined as
	\begin{align*}
		\mathbf{S}_{C,V}=\int_\T^{\oplus}\mathcal{S}_{C,V,\theta} d\theta:\int_\T^{\oplus}\ell^2(\Z,\C^2)d\theta\to\int_\T^{\oplus}\ell^2(\Z,\C^2)d\theta,
	\end{align*}
	and
	$$[\mathbf{S}_{C,V}\boldsymbol{u}](\theta,n)=C\boldsymbol{u}(\theta,n+1)+C^{*}\boldsymbol{u}(\theta,n-1)+V(2n\alpha+\theta)\boldsymbol{u}(\theta,n).$$
	Similarly, the direct integral of $\widehat{S}_{C,V,\theta}$ denoted as $\widehat{\mathbf{S}}_{C,V}$ is
	\begin{equation*}
		[\widehat{\mathbf{S}}_{C,V}\boldsymbol{u}](\theta,n)=\sum_{k\in\mathbb{Z}}V^{(k)}\boldsymbol{u}(\theta,n-k)+\big(Ce^{2\pi i(2n\alpha+\theta)}+C^{*}e^{-2\pi i(2n\alpha+\theta)}\big)\boldsymbol{u}(\theta,n).
	\end{equation*}
	The analyticity and Hermitian symmetry of $V$ imply
	\begin{equation*}
		\sum_{k\in\mathbb{Z}}\|V^{(k)}\|<\infty,\qquad V^{(-k)}=\big(V^{(k)}\big)^{*}.
	\end{equation*}
	Thus the convolution term in the dual operator is bounded and self-adjoint. Its on-site matrix is also self-adjoint, since it has the form $Ce^{2\pi i(2n\alpha+\theta)}+C^{*}e^{-2\pi i(2n\alpha+\theta)}$. Consequently, both $\mathcal{S}_{C,V,\theta}$ and $\widehat{\mathcal{S}}_{C,V,\theta}$ are bounded self-adjoint operators, with bounds uniform in $\theta$. Their direct integrals $\mathbf{S}_{C,V}$ and $\widehat{\mathbf{S}}_{C,V}$ are therefore bounded and self-adjoint on $\mathbb{H}$.

	Define the Aubry duality $\mathbf{A}_{2\alpha}$ on $\mathbb{H}$ as
	\begin{align*}
		[\mathbf{A}_{2\alpha}\boldsymbol{u}](\theta,n)=\sum_{k\in\Z}\int_\T e^{-2\pi i(\theta+2n\alpha)k}e^{-2\pi in\eta}\boldsymbol{u}(\eta,k)d\eta.
	\end{align*}
	Then, for any fixed real analytic $V$ and non-resonant frequency $2\alpha$, the direct integrals $\mathbf{S}_{C,V}$ and $\widehat{\mathbf{S}}_{C,V}$ are unitarily equivalent in the sense that the conjugation $
	\mathbf{A}_{2\alpha}\mathbf{S}_{C,V}=\widehat{\mathbf{S}}_{C,V}\mathbf{A}_{2\alpha} $
	holds.
	Since $\mathbf{A}_{2\alpha}$ is a unitary operator, we have
	\begin{align}\label{5-15-eq7}
		\sigma(\int^\oplus_\T \mathcal{S}_{C,V,\theta}d\theta)=\sigma(\int^\oplus_\T \widehat{\mathcal{S}}_{C,V,\theta}d\theta).
	\end{align}
	
	For the mosaic AMO operator  $\mathcal{H}_\theta$
	defined in \eqref{eq1}, we specialize the strip Jacobi operator to the form:
	\begin{equation*}
		[\mathcal{S}_{\theta} \boldsymbol{u}](n) = C \boldsymbol{u}(n + 1) + C^* \boldsymbol{u}(n - 1) + V(\theta + 2n\alpha) \boldsymbol{u}(n),
	\end{equation*}
	with
	\begin{equation*}
		C = \begin{pmatrix} 0 & 0 \\ 1 & 0 \end{pmatrix}, \quad 
		V(\theta) = \begin{pmatrix} 2\lambda \cos(2\pi \theta) & 1 \\ 1 & 0 \end{pmatrix}.
	\end{equation*}
	The unitary equivalence $\mathcal{P} \mathcal{H}_{\theta} \mathcal{P}^{-1} = \mathcal{S}_{\theta}$ is established via the mapping:
	\begin{align*}
		\mathcal{P}: \ell^2(\mathbb{Z}, \mathbb{C}) &\to \ell^2(\mathbb{Z}, \mathbb{C}^2) \\
		[\mathcal{P}u](n) &:= \begin{pmatrix} u(2n) \\ u(2n+1) \end{pmatrix}
	\end{align*}
	By direct calculations, the dual operator of $\mathcal{S}_{\theta}$ is
	\begin{align}\label{op:dual_S}
		[\widehat{\mathcal{S}}_{\theta}\boldsymbol{u}](n)=\widehat{C}\boldsymbol{u}(n+1)+\widehat{C}\boldsymbol{u}(n-1)+\widehat{V}(2n\alpha+\theta)\boldsymbol{u}(n),
	\end{align}
	\normalsize
	where
	$$\widehat{C}= \begin{pmatrix}
		\lambda&0\\0&0
	\end{pmatrix},\ \widehat{V}(\theta)=\begin{pmatrix}
		0&1+e^{-2\pi i\theta}\\1+e^{2\pi i\theta}&0
	\end{pmatrix}.$$

	The following lemma demonstrates that the spectral structure of $\sigma(\mathcal{H}_\theta)$ is equivalent to that of $\sigma(\widehat{\mathcal{S}}_{\theta})$:
	
	\begin{lemma}[Spectral Equivalence]
		\label{lem:spectral_equivalence}
		There exists a closed set $\Sigma \subseteq \mathbb{R}$ such that for every $\theta \in \mathbb{T}$,
		\[
		\sigma(\mathcal{H}_\theta) = \sigma(\mathcal{S}_{\theta}) = \sigma(\widehat{\mathcal{S}}_{\theta}) = \Sigma.
		\]
	\end{lemma}
	
	\begin{proof}
		The unitary equivalence $\mathcal{P} \mathcal{H}_{\theta} \mathcal{P}^{-1} = \mathcal{S}_{\theta}$ implies:
		\begin{equation}\label{eq:unitary_spectral_equivalence}
			\sigma(\mathcal{H}_\theta) = \sigma(\mathcal{S}_{\theta}) \quad \forall \theta \in \mathbb{T}.
		\end{equation}
		By the direct integral spectral equality in \eqref{5-15-eq7}:
		\begin{equation}\label{eq:direct_integral_equality}
			\sigma\left( \int_{\mathbb{T}}^{\oplus} \mathcal{S}_{\theta}  d\theta \right) = \sigma\left( \int_{\mathbb{T}}^{\oplus} \widehat{\mathcal{S}}_{\theta}  d\theta \right).
		\end{equation}
		Since the base dynamics is minimal, the spectrum of each fiber operator is constant and coincides with the spectrum of the direct integral. Thus, for all $\theta \in \mathbb{T}$:
		\begin{align*}
			\sigma(\mathcal{S}_{\theta}) &= \sigma\left( \int_{\mathbb{T}}^{\oplus} \mathcal{S}_{\theta}  d\theta \right) \\
			\sigma(\widehat{\mathcal{S}}_{\theta}) &= \sigma\left( \int_{\mathbb{T}}^{\oplus} \widehat{\mathcal{S}}_{\theta}  d\theta \right).
		\end{align*}
		Combining \eqref{eq:unitary_spectral_equivalence}, \eqref{eq:direct_integral_equality}, and the minimality condition:
		\[
		\sigma(\mathcal{H}_\theta) = \sigma(\mathcal{S}_{\theta}) = \sigma(\widehat{\mathcal{S}}_{\theta}) =: \Sigma \quad \forall \theta \in \mathbb{T}.
		\]
	\end{proof}

	\section{Integrated density of states and fibered rotation number}
	
	In this section, we establish the relation between the integrated density of states and the fibered rotation number for the singular strip Jacobi  operator $\widehat{\mathcal{S}}_{\theta}$ in \eqref{op:dual_S}, thereby generalizing Theorem \ref{thm:rotation_ids}. This fundamental relation has significant applications in spectral theory, particularly for Cantor spectrum analysis (discussed in Section \ref{Cantor spectrum}).

	We begin by introducing essential concepts and notation. Denote the restriction of $\widehat{\mathcal{S}}_{\theta}$ to $[1,N]$ with Dirichlet boundary conditions by $\widehat{\mathcal{S}}_{\theta}|_{[1,N]}$, i.e.,
	\small
	\begin{align*}
		\widehat{\mathcal{S}}_{\theta}|_{[1,N]}=\begin{pmatrix}
			\widehat V_\theta(1)&\widehat{C}&&&\\\widehat{C}&\widehat V_{\theta}(2)&\widehat{C}&&\\&\ddots&\ddots&\ddots&\\&&\widehat{C}&\widehat V_\theta(N-1)&\widehat{C}\\&&&\widehat{C}&\widehat V_\theta(N)
		\end{pmatrix}_{2N\times 2N},
	\end{align*}
	where $\widehat{V}_\theta(n)=\widehat{V}(2n\alpha+\theta)$.
	\normalsize

	For $\theta \in \mathbb{T}$ and $N \geq 1$, define the empirical spectral measure $d\hat{n}_{\theta,N}$ as the uniform distribution over the ordered eigenvalues $E^{(1)}_{\theta,N} \leq \cdots \leq E^{(2N)}_{\theta,N}$ of $\widehat{\mathcal{S}}_\theta|_{[1,N]}$:
	\begin{align}\label{eq:empirical_measure}
		\int g  d\hat{n}_{\theta,N} := \frac{1}{2N} \sum_{j=1}^{2N} g(E_{\theta,N}^{(j)}) = \frac{1}{2N} \mathrm{Tr}\left(g(\widehat{\mathcal{S}}_{\theta}|_{[1,N]})\right)
	\end{align}
	for bounded measurable $g$.
	
	Define the basis vectors $\boldsymbol{\delta}_0 = \delta_0(1,0)^\top$ and $\boldsymbol{\gamma}_0 = \delta_0(0,1)^\top$, with corresponding spectral measures $\hat{\mu}_{\theta,\boldsymbol{\delta}_0}$ and $\hat{\mu}_{\theta,\boldsymbol{\gamma}_0}$ for $\widehat{\mathcal{S}}_\theta$.
	
	\begin{proposition}[\cite{wang2022, AS}]\label{prop:weak_convergence}
		The following spectral properties hold:
		\begin{enumerate}
			\item For almost every $\theta \in \mathbb{T}$, the empirical measure $d\hat{n}_{\theta,N}$ converges weakly as $N \to \infty$ to the (averaged) spectral measure:
			\begin{equation}\label{eq:limit_measure}
				d\hat{n} := \int_{\mathbb{T}} \frac{1}{2} \left( d\hat{\mu}_{\theta,\boldsymbol{\delta}_0} + d\hat{\mu}_{\theta,\boldsymbol{\gamma}_0} \right) d\theta
			\end{equation}
			
			\item The distribution function
			\begin{align}\label{eq:ids_cumulative}
				\hat{n}_{\theta,N}(E) := \frac{1}{2N} \#\left\{ j : E_{\theta,N}^{(j)} \leq E \right\}
			\end{align}
			converges for a.e. $\theta\in\mathbb{T}$ to a continuous, non-decreasing limit $\hat{n}(E)$ that is $\theta$-independent. This limit defines the integrated density of states (IDS) for $\widehat{\mathcal{S}}_{\theta}$.
			
			\item The IDS $\hat{n}(E)$ is the distribution function of $d\hat{n}$:
			\[
			\hat{n}(E) = \int_{-\infty}^E d\hat{n}(\lambda)
			\]
			and  $\mathrm{supp}(d\hat{n}) = \Sigma.$
		\end{enumerate}
	\end{proposition}
	
	\begin{lemma}\label{lem:duality_ids}
		With the normalizations in \eqref{eq:ids} and \eqref{eq:limit_measure}, the density-of-states measures of the mosaic operator and its dual coincide. In particular, $\hat{n}(E)=n(E)$ for every $E\in\mathbb{R}$.
	\end{lemma}
	
	\begin{proof}
		Let $g\in C_{c}(\mathbb{R})$. By covariance under $T_{\alpha}$ and translation invariance of $\mathrm{d}\theta$, the definition of $n$ gives
		\begin{equation*}
			\int_{\mathbb{R}}g(E)\,\mathrm{d}n(E)=\frac{1}{2}\int_{\mathbb{T}}\Big(\langle e_{0},g(\mathcal{H}_{\theta})e_{0}\rangle+\langle e_{1},g(\mathcal{H}_{\theta})e_{1}\rangle\Big)\,\mathrm{d}\theta.
		\end{equation*}
		Here the two terms represent the two components of $X=\mathbb{T}\times\mathbb{Z}_{2}$ after averaging over $\theta$.
		
		Regard $\boldsymbol{\delta}_{0}$ and $\boldsymbol{\gamma}_{0}$ as elements of $\mathbb{H}$ that are independent of $\theta$. The definition of the Aubry transform gives
		\begin{equation*}
			\mathbf{A}_{2\alpha}\boldsymbol{\delta}_{0}=\boldsymbol{\delta}_{0},\qquad \mathbf{A}_{2\alpha}\boldsymbol{\gamma}_{0}=\boldsymbol{\gamma}_{0}.
		\end{equation*}
		Since $\mathbf{A}_{2\alpha}$ is unitary and intertwines $\mathbf{S}_{C,V}$ and $\widehat{\mathbf{S}}_{C,V}$, functional calculus yields
		\begin{equation*}
			\langle\boldsymbol{u},g(\mathbf{S}_{C,V})\boldsymbol{u}\rangle_{\mathbb{H}}=\langle\boldsymbol{u},g(\widehat{\mathbf{S}}_{C,V})\boldsymbol{u}\rangle_{\mathbb{H}},\qquad \boldsymbol{u}\in\{\boldsymbol{\delta}_{0},\boldsymbol{\gamma}_{0}\}.
		\end{equation*}
		Moreover, $\mathcal{P}\mathcal{H}_{\theta}\mathcal{P}^{-1}=\mathcal{S}_{\theta}$, $\mathcal{P}e_{0}=\boldsymbol{\delta}_{0}$, and $\mathcal{P}e_{1}=\boldsymbol{\gamma}_{0}$. Consequently,
		\begin{equation*}
			\begin{aligned}
				\int_{\mathbb{R}}g(E)\,\mathrm{d}n(E)
				&=\frac{1}{2}\sum_{\boldsymbol{u}\in\{\boldsymbol{\delta}_{0},\boldsymbol{\gamma}_{0}\}}\langle\boldsymbol{u},g(\mathbf{S}_{C,V})\boldsymbol{u}\rangle_{\mathbb{H}}\\
				&=\frac{1}{2}\sum_{\boldsymbol{u}\in\{\boldsymbol{\delta}_{0},\boldsymbol{\gamma}_{0}\}}\langle\boldsymbol{u},g(\widehat{\mathbf{S}}_{C,V})\boldsymbol{u}\rangle_{\mathbb{H}}\\
				&=\int_{\mathbb{R}}g(E)\,\mathrm{d}\hat{n}(E),
			\end{aligned}
		\end{equation*}
		where the last equality follows from \eqref{eq:limit_measure}. Thus $\mathrm{d}n=\mathrm{d}\hat{n}$, and their distribution functions coincide.
	\end{proof}
	
	We next relate the IDS of $\widehat{\mathcal{S}}_{\theta}$ to the fibered rotation number of the associated $\mathrm{SL}(2,\mathbb{R})$ cocycle. The precise formulas are stated in Theorem \ref{11-29-theorem5.21}.

	Let $\boldsymbol{u}(n) = (a(n), b(n))^\mathsf{T}$, then $\widehat{\mathcal{S}}_{\theta}\boldsymbol{u}=z\boldsymbol{u},\ z\neq0$ can be rewritten as 
	\begin{align}\label{eq6}
		\begin{pmatrix}a(n+1)\\ a(n)\end{pmatrix}&=\begin{pmatrix}\frac{z^2-2-2\cos2\pi(2n\alpha+\theta)}{z\lambda}&-1\\1&0\end{pmatrix}\begin{pmatrix}a(n)\\ a(n-1)\end{pmatrix}\\&:=\hat{A}^{z}(\theta+2n\alpha)\begin{pmatrix}a(n)\\a(n-1)\end{pmatrix}.\nonumber
	\end{align}
	Meanwhile, we have
	\begin{align}\label{eq5}
		b(n)=\frac{1+e^{2\pi i(2n\alpha+\theta)}}{z}a(n)\quad\forall n\in\Z.
	\end{align} 
	Thus, spectral analysis of $\widehat{\mathcal{S}}_{\theta}$ reduces to studying the cocycle $({2\alpha}, \hat{A}^E)$ where $\hat A^E(\theta)$ takes the form 
	\begin{align*}
		\hat A^E(\theta)=\begin{pmatrix}\frac{E^2-2-2\cos2\pi\theta}{E\lambda}&-1\\1&0\end{pmatrix}.
	\end{align*}
	If  $E\neq0$, we have  $\hat A^E \in C^\omega(\T, \mathrm{SL}(2,\R))$, thus the fibered rotation number of the cocycle $(2\alpha, \hat A^E)$ is well-defined. Through the Interlacing Lemma (Lemma \ref{11-23-lemma3.15}), we establish the following fundamental relationships between IDS and rotation numbers:
	
	\begin{theorem}\label{11-29-theorem5.21}Let $\alpha\in\mathbb{R}\setminus\mathbb{Q}$. Then the following hold:
		\begin{enumerate}
			\item  For $E>0$, we have
			$\hat{{n}}(E)=1-\rho(2\alpha,\hat A^E).$
			
			\item For $E<0$, we have
			$\hat{{n}}(E)=\frac{1}{2}-\rho(2\alpha,\hat A^E).$
		\end{enumerate}
		Moreover, $\rho(2\alpha,\hat A^{0+})=\frac{1}{2}+\rho(2\alpha,\hat A^{0-})$, and
		\begin{align}
			\rho(2\alpha,\hat A^E)=\left\{\begin{array}{ll}
				2\rho(T_\alpha,S^v_E)  &\text{if } E>0, \\
				2\rho(T_\alpha,S^v_E)-\frac{1}{2}  &\text{if }E<0. 
			\end{array}\right.
		\end{align}
	\end{theorem}	
	Here $\rho(2\alpha,\hat{A}^{0\pm})$ denotes the one-sided limit $\lim_{E\to0^{\pm}}\rho(2\alpha,\hat{A}^{E})$. Similarly, the notation $\rho(2\alpha,\hat{A}^{\pm\infty})$ used below denotes $\lim_{E\to\pm\infty}\rho(2\alpha,\hat{A}^{E})$. We prove this theorem at the end of this section.

	Let $\mathbb{T}_{0}:=\mathbb{T}\setminus\big\{\frac{1}{2}+2k\alpha:k\in\mathbb{Z}\big\},$ where the elements in braces are understood modulo $1$. For every $N\ge2$, $\{E_{N}^{(j)}\}_{j=1}^{2N}$ denote the $2N$ eigenvalues of $\widehat{\mathcal{S}}_{\theta}|_{[1,N]}$ with $\theta\in\T_0$, sorted in ascending order by $j$. Let $\boldsymbol{v}_1,\dots,\boldsymbol{v}_{2N}$ be the orthonormal eigenvectors of $\widehat{\mathcal{S}}_{\theta}|_{[1,N]}$, chosen so that $\widehat{\mathcal{S}}_{\theta}|_{[1,N]}\boldsymbol{v}_j=E^{(j)}_N\boldsymbol{v}_j$. Then, we have the following lemma:
	\begin{lemma}[Interlacing Lemma]\label{11-23-lemma3.15}
		Suppose $\theta\in\T_0$. We have $0\notin\sigma(\widehat{\mathcal{S}}_{\theta}|_{[1,N]})$, and $\langle \boldsymbol{v}_{j},\boldsymbol{\delta}_N\rangle \neq0\ \forall j=1,\dots,2N$. Thus, there exists $m\in[1,2N-1]$ such that $$E_N^{(m)}<0<E_N^{(m+1)}.$$
		Moreover, all the $2N-2$ eigenvalues $\{E_{N-1}^{(j)}\}_{j=1}^{2N-2}$ of $\widehat{\mathcal{S}}_{\theta}|_{[1,N-1]}$ satisfy
		\begin{align*}
			E_N^{(j)}<E_{N-1}^{(j)}<E_N^{({j+1})}
		\end{align*}
		for $1\le j\le m-1$, and
		\begin{align*}
			E_N^{(j)}<E_{N-1}^{(j-1)}<E_N^{({j+1})}
		\end{align*}
		for $m+1\le j\le 2N-1$. 
	\end{lemma}
	
	\begin{proof}
		Since
		\begin{equation}\label{eq:finite-zero-determinant}
			\det\big(\widehat{\mathcal{S}}_{\theta}|_{[1,N]}\big)=(-4)^{N}\prod_{j=1}^{N}\cos^{2}\big(\pi(\theta+2j\alpha)\big)\neq0,\qquad \theta\in\mathbb{T}_{0},
		\end{equation}
		we have $0\notin\sigma(\widehat{\mathcal{S}}_\theta|_{[1,N]})$. The result $\langle \boldsymbol{v}_j, \boldsymbol{\delta}_N \rangle \neq 0$ for all $j = 1, \dots, 2N$ follows from:
		\begin{lemma}\label{lem:endpoint-components}
			Let $\theta\in\mathbb{T}_{0}$, and let $\boldsymbol{u}=(a(j),b(j))_{j=1}^{N}$ be an eigenvector of $\widehat{\mathcal{S}}_{\theta}|_{[1,N]}$ with eigenvalue $z$. Then $z\neq0$ and
			\begin{equation*}
				a(1)a(N)\neq0.
			\end{equation*}
			Moreover, every eigenvalue of $\widehat{\mathcal{S}}_{\theta}|_{[1,N]}$ is simple.
		\end{lemma}
		
		\begin{proof}
			By \eqref{eq:finite-zero-determinant}, one has $z\neq0$. The second component of the eigenvalue equation gives
			\begin{equation*}
				b(j)=\frac{1+e^{2\pi i(2j\alpha+\theta)}}{z}a(j).
			\end{equation*}
			If $a(1)=0$, then $b(1)=0$ since $\theta\in\mathbb{T}_0$. Using the Dirichlet condition $a(0)=0$ in the first component at $j=1$ gives $a(2)=0$. Induction gives $a(j)=b(j)=0$ for every $j$, a contradiction. Thus $a(1)\neq0$. The same argument from the right endpoint gives $a(N)\neq0$. If an eigenspace had dimension at least two, a nontrivial linear combination of two eigenvectors could be chosen with $a(N)=0$, which is impossible. Hence every eigenvalue is simple.
		\end{proof}

		Define a meromorphic function $h$ by
		\begin{align*}
			h(z)=\frac{\det(z-\widehat{\mathcal{S}}_\theta|_{[1,N-1]})}{\det(z-\widehat{\mathcal{S}}_\theta|_{[1,N]})}.
		\end{align*}
		Deleting the row and column corresponding to $\boldsymbol{\delta}_{N}$ from $z-\widehat{\mathcal{S}}_{\theta}|_{[1,N]}$ leaves the second coordinate at site $N$ isolated, with diagonal entry $z$. Thus, Cramer's rule and the spectral decomposition give
		\begin{equation}\label{11-23-eq11}
			\begin{aligned}
				zh(z)
				&=\big\langle\boldsymbol{\delta}_{N},\big(z-\widehat{\mathcal{S}}_{\theta}|_{[1,N]}\big)^{-1}\boldsymbol{\delta}_{N}\big\rangle\\
				&=\sum_{j=1}^{2N}\frac{|\langle\boldsymbol{v}_{j},\boldsymbol{\delta}_{N}\rangle|^{2}}{z-E_{N}^{(j)}},
			\end{aligned}
		\end{equation}
		for every $z\in\mathbb{C}\setminus\sigma(\widehat{\mathcal{S}}_{\theta}|_{[1,N]})$.
		
		Since $0\notin\sigma(\widehat{\mathcal{S}}_{\theta}|_{[1,N]})$, the determinant ratio defining $h$ is analytic at $0$. Evaluating \eqref{11-23-eq11} at $z=0$, we obtain
		\begin{equation*}
			0=-\sum_{j=1}^{2N}\frac{|\langle\boldsymbol{v}_{j},\boldsymbol{\delta}_{N}\rangle|^{2}}{E_{N}^{(j)}}.
		\end{equation*}
		By Lemma \ref{lem:endpoint-components}, all the weights $|\langle\boldsymbol{v}_{j},\boldsymbol{\delta}_{N}\rangle|^{2}$ are strictly positive. Since every eigenvalue is nonzero, the above identity implies that there are both negative and positive eigenvalues. Hence there exists $m\in\{1,\dots,2N-1\}$ such that
		\begin{equation*}
			E_{N}^{(m)}<0<E_{N}^{(m+1)}.
		\end{equation*}
		For each $1\le j\le m-1$, we have
		\begin{align*}
			\lim_{z\downarrow E_{N}^{(j)}}h(z)=-\infty\text{ and }\lim_{z\uparrow E_{N}^{(j+1)}}h(z)=+\infty,
		\end{align*}
		so $h(z)$ vanishes somewhere in the open interval $\left(E_N^{(j)},E_N^{(j+1)}\right)$. It follows that $\widehat{\mathcal{S}}_\theta|_{[1,N-1]}$ has an eigenvalue strictly between $E_N^{(j)}$ and $E_N^{(j+1)}$ for 
		each $1\le j\le m-1$.
		
		Similarly, for each $m+1\le j\le 2N-1$, we have
		\begin{align*}
			\lim_{z\downarrow E_{N}^{(j)}}h(z)=+\infty\text{ and }\lim_{z\uparrow E_{N}^{(j+1)}}h(z)=-\infty,
		\end{align*}
		so $h(z)$ vanishes somewhere in the open interval $\left(E_N^{(j)},E_N^{(j+1)}\right)$. It follows that $\widehat{\mathcal{S}}_\theta|_{[1,N-1]}$ has an eigenvalue strictly between $E_N^{(j)}$ and $E_N^{(j+1)}$ for 
		each $m+1\le j\le 2N-1$. So we have found all the $2N-2$ roots of $h(z)$, and the proof is completed.
	\end{proof}	
	
	Since the restriction $\widehat{\mathcal{S}}_{\theta}|_{[1,n]}$ begins at site $1$, the corresponding transfer matrix starts at phase $\theta+2\alpha$. Thus, for $z\neq0$ and $n\geqslant3$, we have
	\begin{equation}\label{eq:det_representation}
		\hat{A}_{n}^{z}(\theta+2\alpha)=\frac{1}{z^{n}\lambda^{n}}
		\begin{pmatrix}
			\det\big(z-\widehat{\mathcal{S}}_{\theta}|_{[1,n]}\big) & -z\lambda\det\big(z-\widehat{\mathcal{S}}_{\theta}|_{[2,n]}\big)\\
			z\lambda\det\big(z-\widehat{\mathcal{S}}_{\theta}|_{[1,n-1]}\big) & -z^{2}\lambda^{2}\det\big(z-\widehat{\mathcal{S}}_{\theta}|_{[2,n-1]}\big)
		\end{pmatrix}.
	\end{equation}
	Denote $$\Delta_{n,\theta}(E)=\frac{1}{E^n\lambda^n}\det(E-\widehat{\mathcal{S}}_{\theta}|_{[1,n]}),$$ and
	set $\Delta_{0,\theta}(E)=1$ and $\Delta_{-1,\theta}(E)=0$. For $n\geqslant0$, define
	\begin{equation*}
		Z_{n,\theta}(E)=
		\begin{pmatrix}
			\Delta_{n,\theta}(E)\\
			\Delta_{n-1,\theta}(E)
		\end{pmatrix}.
	\end{equation*}
	The determinant recurrence gives
	\begin{equation*}
		Z_{n,\theta}(E)=\hat{A}^{E}(\theta+2n\alpha)Z_{n-1,\theta}(E),\qquad n\geqslant1.
	\end{equation*}
	It is clear that 
	\begin{align*}
		Z_{n,\theta-2\alpha}(E)=\hat A^E_n(\theta)\begin{pmatrix}
			1\\0
		\end{pmatrix}, \qquad  Z_{n-1,\theta}(E)=\hat A^E_{n}(\theta)\begin{pmatrix}
			0\\-1
		\end{pmatrix}.
	\end{align*}
	By \eqref{eq:det_representation}, we have
	\begin{align*}
		\hat{A}_n^E(\theta)=\begin{pmatrix}
			\Delta_{n,\theta-2\alpha}(E)&-\Delta_{n-1,\theta}(E)\\\Delta_{n-1,\theta-2\alpha}(E)&-\Delta_{n-2,\theta}(E)
		\end{pmatrix}
	\end{align*}
	For $E\in\mathbb{R}\setminus\{0\}$, consider the circle action
	\begin{equation*}
		F_{\hat{A}^{E}}(\theta,\psi)=\Big(\theta+2\alpha,\frac{\hat{A}^{E}(\theta)\psi}{\|\hat{A}^{E}(\theta)\psi\|}\Big),\qquad (\theta,\psi)\in\mathbb{T}\times\mathbb{S}^{1}.
	\end{equation*}
	Let $\widetilde{f}_{\hat{A}^{E}}$ be its continuous angle lift normalized by
	\begin{equation*}
		0<\widetilde{f}_{\hat{A}^{E}}(\theta,0)<\frac{1}{2},\qquad \theta\in\mathbb{T}.
	\end{equation*}
	This normalization uniquely determines the lift, which depends continuously on $(E,\theta,t)$ on each component of $\mathbb{R}\setminus\{0\}$. For this family, $\rho(2\alpha,\hat{A}^{E})$ denotes the real rotation number associated with this normalized lift.
	
	Set $\xi_{0,\theta}(E)=0$ and define
	\begin{equation*}
		\xi_{n,\theta}(E)=\widetilde{f}_{\hat{A}^{E}}\big(\theta+2n\alpha,\xi_{n-1,\theta}(E)\big),\qquad n\geqslant1.
	\end{equation*}
	Since $Z_{n,\theta}(E)=\hat{A}_{n}^{E}(\theta+2\alpha)(1,0)^{\mathsf{T}}$, this definition gives
	\begin{equation*}
		\frac{Z_{n,\theta}(E)}{\|Z_{n,\theta}(E)\|}=
		\begin{pmatrix}
			\cos\big(2\pi\xi_{n,\theta}(E)\big)\\
			\sin\big(2\pi\xi_{n,\theta}(E)\big)
		\end{pmatrix}.
	\end{equation*}
	By the definition of the rotation number,
	\begin{equation*}
		\lim_{n\to\infty}\frac{\xi_{n,\theta}(E)}{n}=\rho(2\alpha,\hat{A}^{E}).
	\end{equation*}
	
	For fixed $\theta$, define
	\begin{equation*}
		\zeta_{n}(E)=\frac{\Delta_{n,\theta}(E)+i\Delta_{n-1,\theta}(E)}{\|Z_{n,\theta}(E)\|}=e^{2\pi i\xi_{n,\theta}(E)}.
	\end{equation*}
	For each fixed $n$, these functions are locally real analytic in $E\neq0$. If $E_{1}<E_{2}$ belong to the same component of $\mathbb{R}\setminus\{0\}$, let $\gamma_{E_{1},E_{2}}:[E_{1},E_{2}]\to\mathbb{S}^{1}$ be the path $t\mapsto\zeta_{n}(t)$. Then
	\begin{equation*}
		\xi_{n,\theta}(E_{2})-\xi_{n,\theta}(E_{1})=\frac{1}{2\pi i}\int_{\gamma_{E_{1},E_{2}}}\frac{\mathrm{d}z}{z}=\frac{1}{2\pi i}\int_{E_{1}}^{E_{2}}\frac{\zeta_{n}'(E)}{\zeta_{n}(E)}\,\mathrm{d}E.
	\end{equation*}

	Since IDS and the fibered rotation number are independent of the phase, it suffices to verify the following inequalities on the full-measure subset $\T_0$. We therefore obtain the following lemma.
	\begin{lemma}\label{11-29-lemma5.20}
		(1)  For $0<E_1<E_2$ and $\theta\in\T_0$, we have
		\begin{align*}
			\left|2\xi_{n,\theta}(E_2)-2\xi_{n,\theta}(E_1)+\#\left\{\sigma(\widehat{\mathcal{S}}_{\theta}|_{[1,n]})\cap[E_1,E_2]\right\}\right|\le2,
		\end{align*}
		
		(2) For $E_1<E_2<0$ and $\theta\in\T_0$, we also have
		\begin{align*}
			\left|2\xi_{n,\theta}(E_2)-2\xi_{n,\theta}(E_1)+\#\left\{\sigma(\widehat{\mathcal{S}}_{\theta}|_{[1,n]})\cap[E_1,E_2]\right\}\right|\le2,
		\end{align*}
	\end{lemma}

	\begin{proof}
		Fix $\theta\in\mathbb{T}_{0}$ and let $E_{1}<E_{2}$ belong to the same component of $\mathbb{R}\setminus\{0\}$. Recall that $\Delta_{0,\theta}(E)=1$, $\Delta_{-1,\theta}(E)=0$, and
		\begin{equation*}
			\Delta_{n,\theta}(E)=\frac{E^{2}-2-2\cos\big(2\pi(\theta+2n\alpha)\big)}{E\lambda}\Delta_{n-1,\theta}(E)-\Delta_{n-2,\theta}(E).
		\end{equation*}
		Consecutive terms cannot vanish simultaneously, since the recurrence would then imply $\Delta_{0,\theta}(E)=0$.
		Differentiating the recurrence with respect to $E$ and iterating, we obtain
		\begin{equation*}
			\begin{aligned}
				&\Delta_{n,\theta}'(E)\Delta_{n-1,\theta}(E)-\Delta_{n,\theta}(E)\Delta_{n-1,\theta}'(E)\\
				&\qquad=\frac{1}{\lambda}\sum_{j=1}^{n}\Big(1+\frac{2+2\cos\big(2\pi(\theta+2j\alpha)\big)}{E^{2}}\Big)\Delta_{j-1,\theta}(E)^{2}>0.
			\end{aligned}
		\end{equation*}
		The strict inequality follows from $\lambda>0$ and $\Delta_{0,\theta}(E)=1$. Therefore, the continuous angle lift satisfies
		\begin{equation*}
			\xi_{n,\theta}'(E)=\frac{\Delta_{n,\theta}(E)\Delta_{n-1,\theta}'(E)-\Delta_{n-1,\theta}(E)\Delta_{n,\theta}'(E)}{2\pi\big(\Delta_{n,\theta}(E)^{2}+\Delta_{n-1,\theta}(E)^{2}\big)}<0.
		\end{equation*}
		For $E\neq0$, the definition of $\Delta_{n,\theta}$ and the angle representation give
		\begin{equation*}
			E\in\sigma\big(\widehat{\mathcal{S}}_{\theta}|_{[1,n]}\big)
			\quad\Longleftrightarrow\quad
			\Delta_{n,\theta}(E)=0
			\quad\Longleftrightarrow\quad
			2\xi_{n,\theta}(E)\in\frac{1}{2}+\mathbb{Z}.
		\end{equation*}
		Since $\xi_{n,\theta}$ is continuous and strictly decreasing on $[E_{1},E_{2}]$, each such level is attained exactly once. Hence
		\begin{equation*}
			\begin{aligned}
				&\#\big(\sigma(\widehat{\mathcal{S}}_{\theta}|_{[1,n]})\cap[E_{1},E_{2}]\big)\\
				&\qquad=\#\big\{k\in\mathbb{Z}:2\xi_{n,\theta}(E_{2})\leqslant k+\frac{1}{2}\leqslant2\xi_{n,\theta}(E_{1})\big\}.
			\end{aligned}
		\end{equation*}
		The number of points of $\frac{1}{2}+\mathbb{Z}$ in a closed interval differs from its length by at most $1$. Consequently,
		\begin{equation*}
			\bigg|2\xi_{n,\theta}(E_{2})-2\xi_{n,\theta}(E_{1})+\#\big(\sigma(\widehat{\mathcal{S}}_{\theta}|_{[1,n]})\cap[E_{1},E_{2}]\big)\bigg|\leqslant1\leqslant2.
		\end{equation*}
		This also includes the case where either endpoint is an eigenvalue and proves both parts.
	\end{proof}

	\begin{proof}[Proof of Theorem \ref{11-29-theorem5.21}]
		By Lemma \ref{11-29-lemma5.20}, for $E_1 < E_2 < 0$ or $0 < E_1 < E_2$,
		\begin{align*}
			\lim_{n \to \infty} \left( \frac{1}{n} \xi_{n,\theta}(E_2) - \frac{1}{n} \xi_{n,\theta}(E_1) \right) + \lim_{n \to \infty} \frac{1}{2n} \#\left\{\sigma(\widehat{\mathcal{S}}_{\theta}|_{[1,n]}) \cap [E_1,E_2]\right\} = 0.
		\end{align*}
		Using the definition of the rotation number $\rho(2\alpha, \hat A^E)$ and the integrated density of states $\hat{n}(E)$, we obtain
		\begin{align*}
			\rho(2\alpha, \hat A^{E_2}) - \rho(2\alpha, \hat A^{E_1}) + \hat{n}(E_2) - \hat{n}(E_1) = 0.
		\end{align*}
		
		We now consider two cases:
		\begin{enumerate}
			\item[(1)] For $E_1 < E_2 < 0$: Taking $E_1 \to -\infty$ and using $\rho(2\alpha, \hat A^{-\infty}) = \frac{1}{2}$ yields
			\[
			\hat{n}(E_2) = \frac{1}{2} - \rho(2\alpha, \hat A^{E_2}).
			\]
			
			\item[(2)] For $0 < E_1 < E_2$: Taking $E_2 \to +\infty$ and using $\rho(2\alpha, \hat A^{+\infty}) = 0$ yields
			\[
			\hat{n}(E_1) = 1 - \rho(2\alpha, \hat A^{E_1}).
			\]
		\end{enumerate}
		
		By continuity of $\hat{n}(E)$ at $E = 0$, we equate the expressions:
		\begin{align*}
			1 - \rho(2\alpha, \hat A^{0+}) = \frac{1}{2} - \rho(2\alpha, \hat A^{0-}).
		\end{align*}
		By Lemma \ref{lem:duality_ids} and \eqref{eq:rho_ids}, we have $\hat{n}(E)=n(E)=1-2\rho(T_{\alpha},S_{E}^{v})$. Therefore,
		\begin{align*}
			\rho(2\alpha, \hat A^E) = 
			\begin{cases} 
				\displaystyle
				2\rho(T_\alpha, S^v_E) & \text{for } E > 0, \\[2mm]
				\displaystyle
				2\rho(T_\alpha, S^v_E) - \frac{1}{2} & \text{for } E < 0.
			\end{cases}
		\end{align*}
	\end{proof}

	\section{Reducibility and gap opening}
	
	In this section, we focus on  the subcritical regime of  $ D_E^v$ and $\hat A^E$ and  give reducibility theorems for cocycles  $(2\alpha, D_E^v)$ and $(2\alpha, \hat A^E)$ in subcritical regime, respectively.

	\subsection{Subcritical regime  and Reducibility}

	Recall that for the almost Mathieu operator
	\begin{equation*}
		(H_{\lambda,\alpha,\theta}u)_{n}=u_{n+1}+u_{n-1}+2\lambda\cos\big(2\pi(\theta+n\alpha)\big)u_{n},
	\end{equation*}
	write
	\begin{equation}\label{eq:amo-cocycle-notation}
		S_{E}^{\lambda}(\theta):=\begin{pmatrix}E-2\lambda\cos(2\pi\theta)&-1\\1&0\end{pmatrix},
	\end{equation}
	and denote its phase-independent spectrum by $\Sigma_{\lambda,\alpha}$.
	And  we have the following useful result, which is  presented in \cite[Proposition~5.5]{LYZZ}.
	
	\begin{theorem}\label{thm:amo-resonant-reducibility}
		Let $\alpha\in DC$ and $0<\lambda<1$. For any $E\in\Sigma_{\lambda,\alpha}$ satisfying
		\begin{equation}\label{eq:amo-resonance}
			2\rho\big(\alpha,S_{E}^{\lambda}\big)-k\alpha\in\mathbb{Z}
		\end{equation}
		with $k\in\mathbb{Z}\setminus\{0\}$, and for any
		\begin{equation*}
			0<r<\frac{1}{2\pi}|\ln\lambda|,
		\end{equation*}
		there exist $U\in C_{r}^{\omega}(\mathbb{T},\operatorname{PSL}(2,\mathbb{R}))$ and $\varphi\in\mathbb{R}\setminus\{0\}$ such that
		\begin{equation}\label{eq:amo-parabolic-reduction}
			U(\cdot+\alpha)^{-1}S_{E}^{\lambda}(\cdot)U(\cdot)=\begin{pmatrix}1&\varphi\\0&1\end{pmatrix}.
		\end{equation}
		
	\end{theorem}
	
	Let $E\in\Sigma\setminus\{0\}$ and set
	\begin{equation*}
		R_{E}:=\begin{pmatrix}1&0\\-1&E\end{pmatrix}.
	\end{equation*}
	Then
	\begin{equation}\label{eq:D-exact-amo-conjugacy}
		R_{E}D_{E}^{v}(\theta)R_{E}^{-1}=S_{E^{2}-2}^{\lambda E}(\theta),
	\end{equation}
	and
	\begin{equation}\label{eq:Ahat-exact-amo-identity}
		\hat{A}^{E}(\theta)=S_{(E^{2}-2)/(E\lambda)}^{1/(E\lambda)}(\theta).
	\end{equation}
	For $E\neq0$, these identities identify the two cocycles with almost Mathieu cocycles, up to a constant conjugacy in the case of $(2\alpha,D_{E}^{v})$. We recall the following characterization from \cite{exactm}.
	\begin{lemma}[\cite{exactm}]\label{exactme}
		Let $\alpha \in \mathbb{R} \setminus \mathbb{Q}$ and $E \in \Sigma$. Then:
		\begin{itemize}
			\item The cocycle $(2\alpha, D_E^v)$ is \textit{subcritical} if and only if $|E\lambda| < 1$
			\item The cocycle $(2\alpha, D_E^v)$ is \textit{supercritical} if and only if $|E\lambda| > 1$
		\end{itemize}
	\end{lemma}
	
	For the dual cocycle, define
	\begin{equation*}
		\mathcal{UH}:=\big\{E\in\mathbb{R}\setminus\{0\}:(2\alpha,\hat{A}^{E})\text{ is uniformly hyperbolic}\big\}.
	\end{equation*}
	The following spectral characterization holds.
	\begin{proposition}[\cite{wang2022}]\label{5-11-proposition5.1}
		\begin{equation*}
			\Sigma\setminus\{0\}=(\mathbb{R}\setminus\{0\})\setminus\mathcal{UH}.
		\end{equation*}
	\end{proposition}
	
	\begin{lemma}\label{5-20-lemma4.6}
		Let $\alpha\in\R\backslash\Q$. Then $L(2\alpha,\hat A^E)=\max\{-\ln|E\lambda|,0\}$ for every $E\in\Sigma\backslash\{0\}$. Moreover, for every $E\in\Sigma\backslash\{0\}$, the cocycle $(2\alpha,\hat A^E)$  is
		\begin{itemize}
			\item subcritical iff $|E\lambda|>1$;
			\item supercritical iff $|E\lambda|<1$.
		\end{itemize}
	\end{lemma}

	\begin{remark} We would like to remark that
		Lemma \ref{exactme}  and Lemma \ref{5-20-lemma4.6} assert that the supercritical regime of $\mathcal{H}_\theta$ is consistent with the subcritical regime of $\widehat{\mathcal{S}}_\theta$. 
	\end{remark}
	
	The reducibility results for the cocycles $(2\alpha,D_{E}^{v})$ and $(2\alpha,\hat{A}^{E})$ are obtained by reducing them, up to a constant conjugacy and a phase shift, to almost Mathieu Schr\"odinger cocycles and then applying Theorem~\ref{thm:amo-resonant-reducibility}.

	\begin{theorem}\label{qrr}
		Let $\alpha\in DC$ and
		\begin{equation*}
			E\in\Sigma\cap\big((-\lambda^{-1},0)\cup(0,\lambda^{-1})\big).
		\end{equation*}
		If, for some $k\in\mathbb{Z}\setminus\{0\}$,
		\begin{equation}\label{eq:D-resonance}
			\rho(2\alpha,D_{E}^{v})-k\alpha\in\mathbb{Z},
		\end{equation}
		then there exist $U_{1}\in C^{\omega}(\mathbb{T},\operatorname{PSL}(2,\mathbb{R}))$ and $\varphi_{1}\in\mathbb{R}\setminus\{0\}$ such that
		\begin{equation}\label{eq:D-resonant-parabolic}
			U_{1}(\theta+2\alpha)^{-1}D_{E}^{v}(\theta)U_{1}(\theta)=\begin{pmatrix}1&\varphi_{1}\\0&1\end{pmatrix}.
		\end{equation}
	\end{theorem}
	
	\begin{proof}
		Since $\alpha\in DC$, we also have $2\alpha\in DC$. The energy assumption gives $0<|\lambda E|<1$. The exact conjugacy \eqref{eq:D-exact-amo-conjugacy}, together with the spectral characterization by uniform hyperbolicity, yields $E^{2}-2\in\Sigma_{\lambda E,2\alpha}$.
		
		If $E>0$, the constant projective conjugacy induced by $R_{E}$ preserves orientation, and \eqref{eq:D-resonance} implies
		\begin{equation*}
			2\rho\big(2\alpha,S_{E^{2}-2}^{\lambda E}\big)-k(2\alpha)\in\mathbb{Z}.
		\end{equation*}
		Thus Theorem \ref{thm:amo-resonant-reducibility}, with its $\alpha,\lambda,E$ replaced respectively by $2\alpha,\lambda E,E^{2}-2$, gives a nontrivial parabolic reduction with conjugacy $U$ and parameter $\varphi\neq0$.
		
		If $E<0$, the conjugacy induced by $R_{E}$ reverses orientation, so the same resonance relation holds with label $-k$ and 
		$
		S_{E^{2}-2}^{\lambda E}(\theta)=S_{E^{2}-2}^{-\lambda E}\left(\theta+\frac{1}{2}\right).
		$
		Theorem \ref{thm:amo-resonant-reducibility}, with its $\alpha,\lambda,E,k$ replaced respectively by $2\alpha,-\lambda E,E^{2}-2,-k$, therefore gives a nontrivial parabolic reduction. After the phase translation $\theta\mapsto\theta+\frac{1}{2}$, denote the resulting conjugacy of $S_{E^{2}-2}^{\lambda E}$ again by $U$, and denote its nonzero parabolic parameter by $\varphi$.
		
		In either case, set
		$F_E=\operatorname{diag}(1,\operatorname{sgn}E)$, $\varphi_1=\operatorname{sgn}(E)\varphi$, and
		$$U_1(\theta)=\sqrt{|E|}\,R_E^{-1}U(\theta)F_E\in C^\omega(\mathbb{T},\operatorname{PSL}(2,\mathbb{R})).$$
		Since $\varphi\neq 0$, we have $\varphi_1\neq0$. A direct computation gives
		\[
		U_1(\theta+2\alpha)^{-1}D_E^v(\theta)U_1(\theta)=
		\begin{pmatrix}
			1 & \varphi_1\\
			0 & 1
		\end{pmatrix},
		\]
		which proves the theorem.
	\end{proof}

	In view of the identity \eqref{eq:Ahat-exact-amo-identity}, an argument analogous to the proof of Theorem~\ref{qrr} yields the following result.
	
	\begin{theorem}\label{qrra}
		Let $\alpha\in DC$ and $E\in\Sigma\setminus[-\lambda^{-1},\lambda^{-1}]$. If, for some $k\in\mathbb{Z}\setminus\{0\}$,
		\begin{equation}\label{eq:Ahat-resonance}
			\rho(2\alpha,\hat{A}^{E})-k\alpha\in\mathbb{Z},
		\end{equation}
		then there exist $U_{2}\in C^{\omega}(\mathbb{T},\operatorname{PSL}(2,\mathbb{R}))$ and $\varphi_{2}\in\mathbb{R}\setminus\{0\}$ such that
		\begin{equation}\label{eq:Ahat-resonant-parabolic}
			U_{2}(\theta+2\alpha)^{-1}\hat{A}^{E}(\theta)U_{2}(\theta)=\begin{pmatrix}1&\varphi_{2}\\0&1\end{pmatrix}.
		\end{equation}
	\end{theorem}

	\subsection{\textcolor{black}{Monotonic cocycles and the cone criterion}}
	
	\begingroup
	\color{black}
	In this subsection, we introduce monotonic cocycles and the cone criterion. This part mainly follows Avila and Krikorian \cite{AKMonotonic} and Li, Xu and Zhou \cite{LXZ2026}. Before stating the main result, we introduce the necessary notation and definitions.
	
	Let
	\begin{equation*}
		J=\begin{pmatrix}0&-1\\1&0\end{pmatrix},\qquad \omega(u,v)=u^{\mathsf{T}}Jv.
	\end{equation*}
	Let $\beta\in\mathbb{R}\setminus\mathbb{Q}$, let $I\subseteq\mathbb{R}$ be an interval and let $A_{E}\in C^{1}\big(I,C^{0}\big(\mathbb{T},\mathrm{SL}(2,\mathbb{R})\big)\big)$. For $w\in\mathbb{R}^{2}$, define
	\begin{equation}\label{energy-quadratic-form-LXZ}
		\Psi_{A_{E}}(w,\theta):=\omega\big(A_{E}(\theta)w,\partial_{E}A_{E}(\theta)w\big)=w^{\mathsf{T}}A_{E}(\theta)^{\mathsf{T}}J\partial_{E}A_{E}(\theta)w.
	\end{equation}
	
	\begin{definition}[Parameter monotonicity]\label{def:parameter-monotonicity}
		The family $E\mapsto A_{E}$ is called \emph{positively monotonic} on $I$ if
		\begin{equation*}
			\Psi_{A_{E}}(w,\theta)>0
		\end{equation*}
		for every $E\in I$, $\theta\in\mathbb{T}$ and $w\in\mathbb{R}^{2}\setminus\{0\}$. For $m\geqslant1$, define
		\begin{equation}\label{iterate-family}
			A_{E}^{(m)}(\theta):=A_{E}\big(\theta+(m-1)\beta\big)\cdots A_{E}(\theta+\beta)A_{E}(\theta).
		\end{equation}
		The family $E\mapsto A_{E}$ is called \emph{premonotonic} if $E\mapsto A_{E}^{(s)}$ is positively monotonic for some $s\geqslant1$.
	\end{definition}

	\iffalse
	For $w\in\mathbb{R}^{2}\setminus\{0\}$, denote by $[w]\in\mathbb{RP}^{1}$ its projective class. For $y\in\mathbb{R}$, set
	\begin{equation*}
		\Lambda_{y}:=\begin{pmatrix}\cos(\pi y)\\\sin(\pi y)\end{pmatrix}.
	\end{equation*}
	On the universal cover $\mathbb{R}\times\mathbb{R}$, choose a continuous lift $\Phi_{A_{E}}(\theta,y)$ of the projective action such that
	\begin{equation*}
		[A_{E}(\theta)\Lambda_{y}]=[\Lambda_{y+\Phi_{A_{E}}(\theta,y)}].
	\end{equation*}

	The following result is \cite[Corollary~5.8]{LXZ2026}; see also the parameter monotonicity theory developed in \cite[Section~2]{AKMonotonic}.
	
	\begin{theorem}\label{thm:LXZ-projective-monotonicity}
		Suppose that $E\mapsto A_{E}$ is positively monotonic. Then, for every $\theta\in\mathbb{T}$, $y\in\mathbb{R}$ and $n\geqslant1$, one has
		\begin{equation*}
			\partial_{E}\Phi_{A_{E}^{(n)}}(\theta,y)<0.
		\end{equation*}
	\end{theorem}
	
	\fi

	\begin{lemma}\label{lem:monotone-products-and-iterates}
		The following statements hold.
		\begin{enumerate}
			\item If $E\mapsto F_{E}$ and $E\mapsto G_{E}$ are positively monotonic families in $\mathrm{SL}(2,\mathbb{R})$, then $E\mapsto F_{E}G_{E}$ is positively monotonic.
			\item If $E\mapsto A_{E}^{(s)}$ is positively monotonic, then $E\mapsto A_{E}^{(ms)}$ is positively monotonic for every $m\geqslant1$.
			\item Let $T$ be a homeomorphism of a compact space $X$ and let $A\in C^{0}(X,\mathrm{SL}(2,\mathbb{R}))$. For every fixed $m\geqslant1$,
			\begin{equation}\label{eq:iterate-UH-equivalence}
				(T,A)\text{ is uniformly hyperbolic}\quad\Longleftrightarrow\quad(T^{m},A^{(m)})\text{ is uniformly hyperbolic}.
			\end{equation}
		\end{enumerate}
	\end{lemma}
	
	\begin{proof}
		Since every matrix in $\mathrm{SL}(2,\mathbb{R})$ preserves $\omega$, the product rule gives
		\begin{equation}\label{eq:Psi-product-rule}
			\Psi_{F_{E}G_{E}}(w)=\Psi_{F_{E}}(G_{E}w)+\Psi_{G_{E}}(w).
		\end{equation}
		The first assertion follows immediately, and the second follows by writing $A_{E}^{(ms)}$ as a product of $m$ shifted $s$-step blocks.
		
		If $(T,A)$ is uniformly hyperbolic, then sampling its dynamics every $m$ steps proves that $(T^{m},A^{(m)})$ is uniformly hyperbolic. Conversely, suppose that $(T^{m},A^{(m)})$ is uniformly hyperbolic. The identity
		\begin{equation}\label{eq:sampled-invariance-identity}
			A^{(m)}(Tx)A(x)=A(T^{m}x)A^{(m)}(x)
		\end{equation}
		and uniqueness of the stable and unstable directions of the sampled cocycle show that $A(x)$ maps these directions at $x$ to the corresponding directions at $Tx$. Writing $n=qm+j$, where $0\leqslant j<m$, and using the boundedness of the finitely many $j$-step matrices transfers the exponential estimates to all $n$. This proves \eqref{eq:iterate-UH-equivalence}.
	\end{proof}
	
	\begin{theorem}[Parabolic cone criterion]\label{thm:parabolic-cone-criterion}
		Suppose that $E\mapsto A_{E}$ is premonotonic in a neighborhood of $E_{*}$, and let $s\geqslant1$ be such that $E\mapsto A_{E}^{(s)}$ is positively monotonic. Assume that for some $r\geqslant1$ there exist $B\in C^{0}(\mathbb{T},\operatorname{PSL}(2,\mathbb{R}))$ and $c\neq0$ such that
		\begin{equation}\label{eq:parabolic-reduction}
			B(\theta+r\beta)^{-1}A_{E_{*}}^{(r)}(\theta)B(\theta)=P_{c},\qquad P_{c}:=\begin{pmatrix}1&c\\0&1\end{pmatrix}.
		\end{equation}
		Then there exists $\delta>0$ such that
		\begin{equation}\label{eq:cone-one-sided-UH}
			\begin{cases}
				(\beta,A_{E})\text{ is uniformly hyperbolic for }E\in(E_{*}-\delta,E_{*}),&c>0,\\
				(\beta,A_{E})\text{ is uniformly hyperbolic for }E\in(E_{*},E_{*}+\delta),&c<0.
			\end{cases}
		\end{equation}
		Equivalently, the corresponding $r$-step cocycle is uniformly hyperbolic on the same one-sided interval.
	\end{theorem}
	
	\begin{proof}
		Set $N=rs$. Lemma \ref{lem:monotone-products-and-iterates} implies that $E\mapsto A_{E}^{(N)}$ is positively monotonic. Iterating \eqref{eq:parabolic-reduction} $s$ times gives
		\begin{equation}\label{eq:iterated-parabolic-reduction}
			B(\theta+N\beta)^{-1}A_{E_{*}}^{(N)}(\theta)B(\theta)=P_{c}^{s}=P_{sc}.
		\end{equation}
		Since $B$ takes values in $\operatorname{PSL}(2,\mathbb{R})$, it need not admit a periodic $\mathrm{SL}(2,\mathbb{R})$ representative on $\mathbb{T}$.
		Choose an $\mathrm{SL}(2,\mathbb{R})$ lift $\widetilde{B}$ of $B$ on $\mathbb{R}$. Its possible sign change under $\theta\mapsto\theta+1$ cancels between the two occurrences of $\widetilde{B}$ in
		\begin{equation}\label{eq:conjugated-N-step-family}
			C_{E}(\theta):=\widetilde{B}(\theta+N\beta)^{-1}A_{E}^{(N)}(\theta)\widetilde{B}(\theta),
		\end{equation}
		so $C_{E}$ is a 1-periodic $\mathrm{SL}(2,\mathbb{R})$ cocycle. The relation in \eqref{eq:iterated-parabolic-reduction} may give either $P_{sc}$ or $-P_{sc}$ for this representative. In the latter case, replace the entire family $C_{E}$ by $-C_{E}$. This does not change its projective dynamics, the quadratic form $\Psi_{C_{E}}$, or uniform hyperbolicity. We may therefore assume that
		\begin{equation*}
			C_{E_{*}}(\theta)=P_{a},\qquad a:=sc.
		\end{equation*}
		Because the conjugacy is independent of $E$, symplectic invariance gives
		\begin{equation}\label{eq:Psi-conjugacy-invariance}
			\Psi_{C_{E}}(w,\theta)=\Psi_{A_{E}^{(N)}}(\widetilde{B}(\theta)w,\theta)>0
		\end{equation}
		for every $w\neq0$.
		
		Write $C_{E}(\theta)=(c_{ij}(E,\theta))_{1\leqslant i,j\leqslant2}$. At $E=E_{*}$, one has $c_{11}=c_{22}=1$, $c_{12}=a$ and $c_{21}=0$. With $e_{1}=(1,0)^{\mathsf{T}}$, \eqref{eq:Psi-conjugacy-invariance} gives
		\begin{equation}\label{eq:c21-energy-derivative}
			0<\Psi_{C_{E_{*}}}(e_{1},\theta)=-\partial_{E}c_{21}(E_{*},\theta).
		\end{equation}
		By compactness and continuity, there exist $\kappa>0$ and a neighborhood of $E_{*}$ such that
		\begin{equation}\label{eq:c21-uniform-derivative}
			\partial_{E}c_{21}(E,\theta)\leqslant-\kappa
		\end{equation}
		uniformly in $\theta$.
		
		Suppose first that $c>0$. Then $a=sc>0$, and for $E<E_{*}$ sufficiently close to $E_{*}$, \eqref{eq:c21-uniform-derivative} yields
		\begin{equation}\label{eq:c21-positive-side}
			c_{21}(E,\theta)\geqslant\kappa(E_{*}-E)>0.
		\end{equation}
		After shrinking the energy neighborhood, continuity also gives $c_{11}(E,\theta)>0$, $c_{12}(E,\theta)>0$ and $c_{22}(E,\theta)>0$. Hence $C_{E}(\theta)$ maps every nonzero vector in the closed positive quadrant strictly into its interior, uniformly in $\theta$.
		
		Suppose next that $c<0$. Then $a=sc<0$, and for $E>E_{*}$ sufficiently close to $E_{*}$,
		\begin{equation}\label{eq:c21-negative-side}
			c_{21}(E,\theta)\leqslant-\kappa(E-E_{*})<0.
		\end{equation}
		Moreover, $c_{11}(E,\theta)>0$, $c_{22}(E,\theta)>0$ and $c_{12}(E,\theta)<0$. Thus $C_{E}(\theta)$ maps every nonzero vector in the cone $\{(x,y):x\geqslant0,\ y\leqslant0\}$ strictly into its interior, uniformly in $\theta$.
		
		In either case, $C_{E}$ maps a closed cone strictly into its interior, uniformly in $\theta$. By the standard invariant-cone criterion, $(N\beta,C_{E})$ is uniformly hyperbolic. The family $C_{E}$ is projectively conjugate to $A_{E}^{(N)}$, so $(N\beta,A_{E}^{(N)})$ is uniformly hyperbolic. Lemma \ref{lem:monotone-products-and-iterates} now implies that $(\beta,A_{E})$ is uniformly hyperbolic. Since $a=sc$ has the same sign as $c$, this proves \eqref{eq:cone-one-sided-UH}.
	\end{proof}
	
	We now turn to the cocycle $(2\alpha,D_{E}^{v})$. By verifying the monotonicity condition and applying Theorem \ref{thm:parabolic-cone-criterion}, we obtain the following result.
	
	\begin{theorem}\label{thm:D-parabolic-open-gap}
		Let $E_{*}\in\Sigma\cap\big((-\lambda^{-1},0)\cup(0,\lambda^{-1})\big)$. Suppose that there exist an analytic map $U_{1}:\mathbb{T}\to\operatorname{PSL}(2,\mathbb{R})$ and $\varphi_{1}\neq0$ such that
		\begin{equation}\label{eq:D-parabolic-reduction}
			U_{1}(\theta+2\alpha)^{-1}D_{E_{*}}^{v}(\theta)U_{1}(\theta)=\begin{pmatrix}1&\varphi_{1}\\0&1\end{pmatrix}.
		\end{equation}
		Then there exists $\delta>0$ such that
		\begin{equation}\label{eq:D-one-sided-resolvent}
			I_{*}:=\begin{cases}
				(E_{*}-\delta,E_{*}),&\varphi_{1}>0,\\
				(E_{*},E_{*}+\delta),&\varphi_{1}<0
			\end{cases}
			\subseteq\mathbb{R}\setminus\Sigma.
		\end{equation}
		In particular, if $E_{*}$ lies strictly between two points of $\Sigma$, then $E_{*}$ is a boundary point of a non-collapsed spectral gap.
	\end{theorem}
	
	\begin{proof}
		Recall that
		\begin{equation*}
			D_{E}^{v}(\theta)=\begin{pmatrix}E&-1\\1&0\end{pmatrix}\begin{pmatrix}E-2\lambda\cos(2\pi\theta)&-1\\1&0\end{pmatrix}.
		\end{equation*}
		For $w=(x,y)^{\mathsf{T}}\in\mathbb{R}^{2}\setminus\{0\}$, a direct calculation gives
		\begin{equation*}
			\Psi_{D_{E}^{v}}(w,\theta)=x^{2}+\big((E-2\lambda\cos(2\pi\theta))x-y\big)^{2}>0.
		\end{equation*}
		Thus $E\mapsto D_{E}^{v}$ is positively monotonic. In this application, the monotonicity witness and the parabolic reduction step are $s=r=1$. Theorem \ref{thm:parabolic-cone-criterion} therefore shows that $(2\alpha,D_{E}^{v})$ is uniformly hyperbolic for every $E\in I_{*}$.
		
		Let $X_{j}:=\mathbb{T}\times\{j\}$ for $j\in\mathbb{Z}_{2}$. Lemma \ref{lem:monotone-products-and-iterates} does not apply directly here: it compares a cocycle with its iterate on the same full base, whereas $(2\alpha,D_{E}^{v})$ is only the restriction of the two-step cocycle $(T_{\alpha}^{2},(S_{E}^{v})^{(2)})$ to $X_{0}$. Since $T_{\alpha}$ exchanges $X_{0}$ and $X_{1}$, for $x\in X_{0}$ the identity
		\begin{equation}\label{eq:mosaic-return-identity}
			(S_{E}^{v})^{(2)}(T_{\alpha}x)S_{E}^{v}(x)=S_{E}^{v}(T_{\alpha}^{2}x)(S_{E}^{v})^{(2)}(x)
		\end{equation}
		extends the invariant stable and unstable bundles, together with their exponential estimates, from $X_{0}$ to $X_{1}$. Hence the full two-step cocycle $(T_{\alpha}^{2},(S_{E}^{v})^{(2)})$ is uniformly hyperbolic. Applying Lemma \ref{lem:monotone-products-and-iterates} to this full-base cocycle shows that $(T_{\alpha},S_{E}^{v})$ is uniformly hyperbolic. Johnson's theorem \cite{johnson1986exponential} now gives $I_{*}\subseteq\mathbb{R}\setminus\Sigma$. Since $E_{*}\in\Sigma$, the one-sided interval $I_{*}$ lies in a component of $\mathbb{R}\setminus\Sigma$ having $E_{*}$ as an endpoint. If $\inf\Sigma<E_{*}<\sup\Sigma$, this component is bounded, which proves the final assertion.
	\end{proof}
	
	We next apply the same criterion to the dual cocycle $(2\alpha,\hat{A}^{E})$.
	
	\begin{theorem}\label{thm:Ahat-parabolic-open-gap}
		Let $E_{*}\in\Sigma\setminus[-\lambda^{-1},\lambda^{-1}]$. Suppose that there exist an analytic map $U_{2}:\mathbb{T}\to\operatorname{PSL}(2,\mathbb{R})$ and $\varphi_{2}\neq0$ such that
		\begin{equation}\label{eq:Ahat-parabolic-reduction}
			U_{2}(\theta+2\alpha)^{-1}\hat{A}^{E_{*}}(\theta)U_{2}(\theta)=\begin{pmatrix}1&\varphi_{2}\\0&1\end{pmatrix}.
		\end{equation}
		Then there exists $\delta>0$ such that
		\begin{equation}\label{eq:Ahat-one-sided-resolvent}
			\begin{cases}
				(E_{*}-\delta,E_{*})\subseteq\mathbb{R}\setminus\Sigma,&\varphi_{2}>0,\\
				(E_{*},E_{*}+\delta)\subseteq\mathbb{R}\setminus\Sigma,&\varphi_{2}<0.
			\end{cases}
		\end{equation}
		If $E_{*}$ lies strictly between two points of $\Sigma$, then $E_{*}$ is a boundary point of a non-collapsed spectral gap.
	\end{theorem}
	
	\begin{proof}
		Write
		\begin{equation*}
			\hat{A}^{E}(\theta)=\begin{pmatrix}q_{E}(\theta)&-1\\1&0\end{pmatrix},\qquad q_{E}(\theta)=\frac{E^{2}-2-2\cos(2\pi\theta)}{E\lambda}.
		\end{equation*}
		Since $\lambda>0$ and $E\neq0$, one has
		\begin{equation*}
			\partial_{E}q_{E}(\theta)=\frac{1}{\lambda}\bigg(1+\frac{2+2\cos(2\pi\theta)}{E^{2}}\bigg)\geqslant\frac{1}{\lambda}>0.
		\end{equation*}
		For $w=(x,y)^{\mathsf{T}}\in\mathbb{R}^{2}\setminus\{0\}$, symplectic invariance and the chain rule give
		\begin{equation*}
			\Psi_{(\hat{A}^{E})^{(2)}}(w,\theta)=\partial_{E}q_{E}(\theta)x^{2}+\partial_{E}q_{E}(\theta+2\alpha)\big(q_{E}(\theta)x-y\big)^{2}>0.
		\end{equation*}
		Hence $E\mapsto\hat{A}^{E}$ is premonotonic with witness $s=2$. The parabolic reduction \eqref{eq:Ahat-parabolic-reduction} is a one-step reduction, so $r=1$ and the common iterate in Theorem \ref{thm:parabolic-cone-criterion} is $N=rs=2$. Indeed, taking the two-step iterate of \eqref{eq:Ahat-parabolic-reduction} gives
		\begin{equation*}
			U_{2}(\theta+4\alpha)^{-1}(\hat{A}^{E_{*}})^{(2)}(\theta)U_{2}(\theta)=\begin{pmatrix}1&2\varphi_{2}\\0&1\end{pmatrix}.
		\end{equation*}
		Since $2\varphi_{2}$ has the same sign as $\varphi_{2}$, Theorem \ref{thm:parabolic-cone-criterion} gives uniform hyperbolicity of $(2\alpha,\hat{A}^{E})$ on the corresponding one-sided interval. Proposition \ref{5-11-proposition5.1} gives \eqref{eq:Ahat-one-sided-resolvent}. The conclusion follows.
	\end{proof}
	
	\begin{corollary}\label{parabolic-implies-open-gap}
		The nontrivial parabolic reductions in Theorems \ref{qrr} and \ref{qrra} produce a nonempty one-sided interval in the resolvent set. Whenever the corresponding resonant energy lies strictly between two points of $\Sigma$, it is therefore a boundary point of a non-collapsed open gap.
	\end{corollary}

	\endgroup
	
	\section{{Proof of Cantor spectrum}}\label{Cantor spectrum}

	We now complete the proof of Theorem~\ref{5-15-theorem1.1}. Let $\Sigma$ denote the spectrum of the family of operators defined by \eqref{eq1}. Since these operators are bounded and self-adjoint, $\Sigma$ is nonempty and compact. By Pastur's theorem \cite[Theorem~4.2.4]{DF2022}, the almost-sure spectrum of an ergodic Schr\"odinger family has no isolated points. Therefore, it remains only to prove that $\operatorname{int}\Sigma=\varnothing$. 
	
	Suppose, for a contradiction, that $\operatorname{int}\Sigma\neq\varnothing$.
	After deleting the three energies $-\lambda^{-1}$, $0$ and $\lambda^{-1}$ from a nonempty open subinterval, choose a nondegenerate compact interval $J=[a,b]\subseteq\operatorname{int}\Sigma$ contained in exactly one of
	\begin{equation}\label{eq:four-energy-regions}
		(-\infty,-\lambda^{-1}),\qquad(-\lambda^{-1},0),\qquad(0,\lambda^{-1}),\qquad(\lambda^{-1},\infty).
	\end{equation}
	Note that the integrated density of states whose density-of-states measure has support $\Sigma$ is strictly increasing on $J$. Indeed, if $a\leqslant E_{1}<E_{2}\leqslant b$, then any nonempty open subinterval of $(E_{1},E_{2})$ meets the support and therefore has positive density-of-states measure. The relevant integrated density of states is also continuous, so the corresponding circle-valued rotation number admits a continuous real lift on $J$.
	
	Assume first that $J$ lies in one of the two inner intervals in \eqref{eq:four-energy-regions}. Choose a continuous lift $R(E)$ of $\rho(2\alpha,D_{E}^{v})$. The iterate identity for the rotation number and \eqref{eq:rho_ids} give, for some fixed $m\in\mathbb{Z}$,
	\begin{equation}\label{eq:D-lift-via-ids}
		R(E)=1-n(E)+m,\qquad E\in J.
	\end{equation}
	Thus $R$ is continuous and strictly decreasing. Since $\alpha$ is irrational, the set
	\begin{equation*}
		\{k\alpha\bmod\mathbb{Z}:k\in\mathbb{Z}\setminus\{0\}\}
	\end{equation*}
	is dense in $\mathbb{R}/\mathbb{Z}$. Hence there exist $k\in\mathbb{Z}\setminus\{0\}$ and $\ell\in\mathbb{Z}$ such that
	\begin{equation*}
		R(b)<k\alpha+\ell<R(a).
	\end{equation*}
	The intermediate value theorem gives $E_{k}\in(a,b)$ satisfying
	\begin{equation*}
		\rho(2\alpha,D_{E_{k}}^{v})-k\alpha\in\mathbb{Z}.
	\end{equation*}
	Theorem \ref{qrr} gives a nontrivial parabolic reduction, and Theorem \ref{thm:D-parabolic-open-gap} gives a one-sided resolvent interval adjacent to $E_{k}$. Since $E_{k}\in(a,b)$, after decreasing its length if necessary, this resolvent interval is contained in $(a,b)\subseteq\Sigma$, a contradiction.
	
	Now suppose that $J$ lies in one of the two outer intervals in \eqref{eq:four-energy-regions}. Choose a continuous lift $R(E)$ of $\rho(2\alpha,\hat{A}^{E})$. Theorem \ref{11-29-theorem5.21} gives, for some fixed $m\in\mathbb{Z}$,
	\begin{equation}\label{eq:Ahat-lift-via-ids}
		R(E)=\begin{cases}
			\displaystyle\frac{1}{2}-\widehat{n}(E)+m,&J\subseteq(-\infty,-\lambda^{-1}),\\[2mm]
			\displaystyle1-\widehat{n}(E)+m,&J\subseteq(\lambda^{-1},\infty).
		\end{cases}
	\end{equation}
	By Proposition \ref{prop:weak_convergence}, $\operatorname{supp}(\mathrm{d}\widehat{n})=\Sigma$, so $R$ is again continuous and strictly decreasing. Density of the nonzero labels gives $k\in\mathbb{Z}\setminus\{0\}$ and $\ell\in\mathbb{Z}$ such that
	\begin{equation*}
		R(b)<k\alpha+\ell<R(a).
	\end{equation*}
	The intermediate value theorem gives $E_{k}\in(a,b)$ satisfying
	\begin{equation*}
		\rho(2\alpha,\hat{A}^{E_{k}})-k\alpha\in\mathbb{Z}.
	\end{equation*}
	Theorems \ref{qrra} and \ref{thm:Ahat-parabolic-open-gap} produce a one-sided resolvent interval adjacent to $E_{k}$. After decreasing its length if necessary, this interval is contained in $(a,b)\subseteq\Sigma$, again a contradiction. Therefore $\operatorname{int}\Sigma=\varnothing$.
	Thus $\Sigma$ is nonempty, compact, perfect and nowhere dense, and hence it is a Cantor set.

	\section*{Acknowledgments}
	Jiawei He is supported by the National Natural Science Foundation of China (Grant No. 12501247), the Natural Science Foundation of Fujian Province (Grant No. 2025J08232), the Startup Fund for Advanced Talents of Putian University (Grant No. 2023120), and the Fujian Alliance of Mathematics (Grant No. 2025SXLMMS07). Yuan Shan is supported by the National Natural Science Foundation of China (Grant No. 12671214). Yongjian Wang is supported by the National Natural Science Foundation of China (Grant No. 12401208) and the Natural Science Foundation of Jiangsu Province (Grant No. BK20241431).

	\section*{Statements and Declarations}
	{\bf Conflict of Interest} 
	The authors declare no conflicts of interest.
	
	\vspace{0.2in}
	{\bf Data Availability}
	Data sharing is not applicable to this article as no new data were created or analyzed in this study.

\end{document}